\documentclass[12pt]{article}
\usepackage[margin=1in]{geometry}
\usepackage[small]{titlesec}
\usepackage{natbib,amsmath,amssymb,amsthm,graphicx,setspace,paralist,booktabs,rotating,times}
\usepackage{hyperref}
\usepackage{subfigure}
\usepackage{cases}
\usepackage[section]{placeins}
\usepackage[utf8]{inputenc}
\DeclareUnicodeCharacter{2061}{}
\AtBeginDocument{%
	\abovedisplayskip=5pt plus 2pt minus 2pt
	\belowdisplayskip=\abovedisplayskip
	\abovedisplayshortskip=2pt plus 2pt minus 2pt
	\belowdisplayshortskip=\belowdisplayskip}
\titlespacing*{\section}{0pt}{*2}{*1}
\titlespacing*{\subsection}{0pt}{*2}{*1}
\bibpunct{(}{)}{;}{a}{}{,}
\newtheorem{theorem}{Theorem}
\newtheorem{lemma}{Lemma}

\theoremstyle{definition}
\newtheorem{condition}{Condition}

\DeclareMathOperator*\argmin{\arg\min}

\def\bX{\mathbf{X}}
\def\bY{\mathbf{Y}}

\def\bmu{\boldsymbol{\mu}}

\allowdisplaybreaks[1]

\begin{document}	
\begin{titlepage}
	\setstretch{1.24}
	\title{Gene regulatory network in single cells based on the Poisson log-normal model}
	\author{Feiyi Xiao$^1$, Junjie Tang$^1$, Huaying Fang$^2$ and Ruibin Xi$^{1,*}$}
	\date{}
	\maketitle
	\thispagestyle{empty}
		
	\footnotetext{
		\hspace{-6mm}$^1$ School of Mathematical Sciences, Peking University, Beijing, China \\
		$^2$ Beijing Advanced Innovation Center for Imaging Theory and Technology,
		Capital Normal University, Beijing, China\\
		$^*$ To whom correspondence should be addressed. Email: ruibinxi@math.pku.edu.cn.
	}

\begin{abstract}
Gene regulatory network inference is crucial for understanding the complex molecular interactions in various genetic and environmental conditions. The rapid development of single-cell RNA sequencing (scRNA-seq) technologies unprecedentedly enables gene regulatory networks inference at the single cell resolution. However, traditional graphical models for continuous data, such as Gaussian graphical models, are inappropriate for network inference of scRNA-seq's count data. Here, we model the scRNA-seq data using the multivariate Poisson log-normal (PLN) distribution and represent the precision matrix of the latent normal distribution as the regulatory network. We propose to first estimate the latent covariance matrix using a moment estimator and then estimate the precision matrix by minimizing the lasso-penalized D-trace loss function. We establish the convergence rate of the covariance matrix estimator and further establish the convergence rates and the sign consistency of the proposed PLNet estimator of the precision matrix in the high dimensional setting. The performance of PLNet is evaluated and compared with available methods using simulation and gene regulatory network analysis of scRNA-seq data.
	 
	\end{abstract}
\end{titlepage}
	
\section{Introduction}
The rapid development of single cell RNA sequencing (scRNA-seq) technologies has provided tremendous opportunities for understanding transcriptional states and activities of single cells. scRNA-seq can be used to unveil the cellular diversity in various biological conditions, identify new cell types and trace trajectories of cell lineages in development. Especially, scRNA-seq allows uncovering gene regulatory networks at single cell level. Gaussian graphical model (GGM) \citep{Meinshausen2006,Yuan2007} is widely used for gene regulatory network analysis \citep{yin2011sparse,ma2007arabidopsis,wille2004sparse}. GGM assumes that each sample is drawn from a multivariate Gaussian distribution, in which the precision matrix (i.e., the inverse of the covariance matrix) represents the gene regulatory network. In GGM, two nodes are not connected if the corresponding random variables are conditionally independent given all other variables, or equivalently if the corresponding element in the precision matrix is zero. However, gene expressions from scRNA-seq are count data and the Gaussian assumption is inappropriate. In particular, for the recent unique molecular identifier (UMI) based scRNA-seq data \citep{farrell2018single,zheng2017massively}, the expression counts are rather small and contain many zeros (the dropout problem). Transformations like taking logarithms cannot make the Gaussian model a good approximation to the distribution of scRNA-seq expression data and may distort the correlation structure when there are many 0's in scRNA-seq data.


Poisson distribution is a natural choice for modeling count data. Researchers have generalized the univariate Poisson distribution to multivariate distributions for network analysis with count data \citep{Inouye2017}. One such model is the Poisson graphical model \citep{Allen2013}, which assumes that each node has a conditional univariate Poisson distribution, but this multivariate joint Poisson distributions are not flexible enough and can only capture negative dependencies. To circumvent this limitation, \citet{yang2013poisson} proposed extensions of the Poisson graphical model. However, scRNA-seq data are not exact gene expressions in single cells but are measurements (with noises) of the gene expressions. The regulatory network is about the inter-dependency between gene expresssions, but these generalized Poisson graphical models directly impose graph structure on the count data and the technical noises are also involved in the network. In addition, the over-dispersion in scRNA-seq indicates that modeling scRNA-seq counts by Poisson distributions may be inadequate.  Negative binomial distributions are often used to account for the over-dispersion \citep{Robinson2009}, but it is more difficult to generalize negative bionomial distributions to describe the network structure in the multivarate count data. Here, we propose to use the multivariate Poisson log-normal (PLN) distribution for gene regulatory network analysis based on scRNA-seq data. The PLN distribution is a mixture of Poisson and multivariate log-normal distributions \citep{AITCHISON1989}. A random vector $\bY = (Y_1,\cdots,Y_p)^T \in \mathbb{R}^p$ is from a PLN distribution, if conditional on a latent variable $\bX = (X_1,\cdots,X_p)^T \in \mathbb{R}^p$ with $\log(\bX) \sim \mbox{N}(\bmu,\Sigma)$, $\bY$ follows a multivariate Poisson distribution $\prod_{j=1}^p\mbox{Poisson}(X_j)$.

Like negative binomial distributions, PLN distributions have over-dispersion and hence are more suitable for modeling scRNA-seq data than Poisson distributions \citep{Inouye2017}. The major advantage of the PLN model is that, similar to the GGM, the precision matrix $\Theta = \Sigma^{-1}$ of the latent log-normal vector $\bX$ can represent the gene regulatory network. Network recovery of single cells can be achieved by estimating precision matrices of PLN models. Previous experimental researches showed that gene expressions of single cells follow log-normal distributions \citep{bengtsson2005gene}. The PLN model of scRNA-seq data thus has the following explanation. The latent variables $X_1,\cdots,X_p$ are the true expressions of $p$ genes in a single cell. The logarithms of the expressions are jointly normally distributed and the precision matrix captures the gene-gene interactions in single cells. The counts $Y_1,\cdots,Y_p$ are the measurements of the true expressions $X_1,\cdots,X_p$. It is thus reasonable to assume that $Y_1,\cdots,Y_p$ are conditionally independent and their conditional expectations depend on the true expressions. Largely speaking, the log-normal layer of the PLN captures the biological fluctuation of gene expressions and the Poisson layer accounts for the technical and measurement noises. Only the biological fluctuation reflects gene-gene interactions and the regulatory network is the precision matrix of the latent log-normal model.

Based on the sparsity assumption, researchers have proposed many methods for estimating the precision matrix of the GGM and established consistency theories, such as methods by maximizing the penalized log-likelihood \citep{Yuan2007,Friedman2008}, by solving an equivalent regression problem with lasso penalty \citep{Meinshausen2006,Peng2009} or by minimizing a smooth convex loss function (called D-trace loss) with lasso penalty \citep{Zhang2014}. A few algorithms have been developed for estimating the precision matrix of the PLN model in high-dimensional settings. Compared with GGM, the likelihood of the PLN model is more complicated since it involves a multivariate integration and does not have a close form. Maximizing the penalized log-likelihood of the PLN model is more difficult. \cite{Wu2018} proposed to use Monte Carlo to approximate the log-likelihood and estimate the precision matrix by maximizing the penalized approximated log-likelihood of the PLN model, while \cite{Chiquet2019} developed a computational more appealing method based on the variational approximation. However, these methods are all based on approximation of the log-likelihood and the accuracies of these approximations need to be further elaborated. More importantly, no convergence theory has been developed for these precision matrix estimators.

In this article, we propose to first estimate the covariance matrix $\Sigma$ of the latent log-normal variables in the PLN using a moment estimator $\hat{\Sigma}$, and then estimate the precision matrix by minimizing the lasso penalized D-trace loss \citep{Zhang2014}. One advantage of this moment-based approach (called PLNet) is that it avoids computing the log-likelihood of the PLN model. Minimizing the penalized D-trace loss is computationally cost-effective. Thus, this estimator is generally computationally more efficient. Furthermore, we show that, under an irrepresentability condition and a few other mild conditions, the estimator given by PLNet is a consistent estimator of the precision matrix in high dimensional settings. Comprehensive simulation analyses show that PLNet provides more accurate estimates and is computationally more efficient than available methods. We also demonstrate the application of PLNet to scRNA-seq data.

\section{Model}
\label{sec2}

\subsection{The PLN graphical model}
Let $\bY_i = (Y_{i1},\cdots,Y_{ip})^T$ be the observed count data of the $i$th sample and $\bX_i = (X_{i1},\cdots,X_{ip})^T$ be the latent random vector. In scRNA-seq data, $Y_{ij}$ and $X_{ij}$ are the observed expression and the underlying ``true" expressions of the $j$th gene in the $i$th cell, respectively. We assume that conditional on the latent random vector $\bX_i$, $Y_{ij}$'s are independent Poisson random variables with the mean parameters  $S_iX_{ij}$ ($j=1,\cdots,p$), where $S_i$ a known scaling factor. In scRNA-seq data, $S_i$ corresponds to the library size $S_i$ of the $i$th cell. The library size is related to the total sequencing reads and can be estimated by the sum of counts within each cell or by other available methods \citep{Love2014,Lun2016,Vallejos2015}. The latent random vector $\bX_i$ follows a multivariate log-normal random vector with mean $\bmu$ and covariance $\Sigma$. The precision matrix $\Theta=\Sigma^{-1}$ represents the network. In summary, we have the following graphical model for count data, for $1\leq i\leq n$,
\begin{equation}\label{equ1}
\begin{aligned}
\bY_i |\bX_i &\sim \prod_{j=1}^p\mbox{Poisson}(S_iX_{ij}), \\
\log(\bX_i)&\sim \mbox{N}\left( \bmu ,\Theta^{-1} \right).
\end{aligned}
\end{equation}

The above PLN model is a little different from the classical form \citep{AITCHISON1989}, in which $S_i=1$ for all $1\leq i \leq n$ . We assume that the network $\Theta$ is sparse and therefore could use the penalized log-likelihood to estimate $\Theta$. However, the likelihood function in the PLN model involves a $p$-dimensional integration and is difficult to compute, easpecially when $p$ is large. \cite{Chiquet2019} developed a variational algorithm called Variational inference for PLN model (VPLN) to maximize the penalized log-likelihood. Although the variational method is computationally more feasible than directly maximizing the penalized log-likelihood, the estimator's theoretical properties are difficult to obtain. We instead develop an estimator using the moment method that is computationally efficient and has good theoretical properties.

\subsection{A moment based estimator}
Based on the moment method and the D-trace method, we propose a two-step method called PLNet to estimate the sparse precision matrix $\Theta$. We first use the moment method to estimate $\Sigma$ with a semi-positive definite estimator $\hat{\Sigma}$. Then, we apply the D-trace method \citep{Zhang2014} to the covariance estimator $\hat{\Sigma}$ to estimate the sparse precision matrix $\Theta$. We show that the derived estimator $\hat{\Theta}$ is a consistent estimator of $\Theta$ even when the dimensionality diverges to infinite with the sample size going to infinity.

Let $\bmu=\left[ \mu_i \right]_{1\leq i\leq p}$ be the mean vector, $\Sigma=\left[ \sigma_{ij} \right]_{1\leq i,j\leq p}$ be the covariance matrix and $\alpha_i=\mu_i+\sigma_{ii}/2$ for $ 1 \leq i \leq p$. From the first two moments of the PLN distribution, we have
\begin{equation}\label{equ2}
\begin{aligned}
E\left(Y_{ij}/S_i\right)&=\alpha_j,\\
E\left(\left(Y_{ij}^2-Y_{ij}\right)/S_i^2\right)&=\alpha_j^2{\rm exp}\left( \sigma_{jj}\right), \\
E\left(Y_{ij}Y_{ik}/S_i^2 \right)&=\alpha_j\alpha_{k}{\rm exp}\left( \sigma_{jk}\right),
\end{aligned}
\end{equation}
where $1 \leq i \leq n$ and $1\leq j \neq k \leq p$. Let $\tilde{\alpha}_j=n^{-1} \sum_{i=1}^n Y_{ij}/S_i$  for $1\leq j \leq p$. Then, a candidate moment estimator $\tilde{\Sigma}=\left[ \tilde{\sigma}_{ij}\right] _{1\le i,j\le p}$ for the covariance matrix is
\begin{equation}\label{equ3}
\tilde{\sigma}_{jk} = \begin{cases}
\log⁡\left( n^{-1} \sum_{i=1}^{n}\left[ \left\lbrace Y_{ij}\left( Y_{ij}-1\right) \right\rbrace / S_i^2 \right] \right) - 2\log⁡\left( \tilde{\alpha}_j\right),\ &{\rm for}\  1\leq j=k \leq p, \\
\log⁡\left[ n^{-1} \sum_{i=1}^{n}\left\lbrace \left( Y_{ij}Y_{ik} \right)/S_i^2 \right\rbrace\right]   - \left\lbrace  \log⁡\left( \tilde{\alpha}_j\right)+\log⁡\left( \tilde{\alpha}_k\right)\right\rbrace ,\ &{\rm for}\  1\leq j\neq k \leq p.
\end{cases}
\end{equation}

The above moment estimator $\tilde{\Sigma}$ maybe not semi-positive definite. However, the D-trace method requires the input covariance matrix estimator to be semi-positive definite to guarantee the convexity of the loss function. To ensure semi-positive definiteness, we project $\tilde{\Sigma}$ to the space of semi-positive definite matrices and identify $\check{\Sigma}$ that is closest to $\tilde{\Sigma}$ in the space, i.e.,
\begin{equation}
\check{\Sigma}=\argmin_{A\succeq 0}\left\|A-\tilde{\Sigma}\right\|_{\infty},
\end{equation}
where $A\succeq 0$ means $A$ is a semi-positive definite matrix, and $\left\|A\right\|_{\infty}={\rm max}_{i,j}\left| A_{ij}\right| $ is the element-wise $l_\infty$-norm of the matrix $A$. The optimization problem for $\Check{\Sigma}$ can be solved by a splitting conic solver \citep{Fu2020}. Using $\Check{\Sigma}$ in the D-trace loss
\begin{equation}\label{ea}
\check{\Theta}=\argmin_{\Theta\succeq 0}\frac{1}{2}{\rm tr}\left( \check{\Sigma}\Theta^2\right) -{\rm tr}\left( \Theta \right) +\lambda_n\left\|\Theta\right\|_{1,\text{off}},
\end{equation}
we renders a consistent estimator $\check{\Theta}$ of the precision matrix. However, we find that minimizing the penalized D-trace loss with this covariance estimator $\check{\Sigma}$ can be computationally expensive in many scenarios. Therefore, we propose to use the following estimator to estimate $\Sigma$,
\begin{equation}
\hat{\Sigma}=\check{\Sigma}+\left\|\check{\Sigma}-\tilde{\Sigma}\right\|_{\infty}I_p,
\end{equation}
where $I_p$ is the $p\times p$ identity matrix. With the covariance matrix estimator $\hat{\Sigma}$, we apply the D-trace method to estimate the precision matrix,
\begin{equation}\label{eaa}
\hat{\Theta}=\argmin_{\Theta\succeq 0}\frac{1}{2}{\rm tr}\left( \hat{\Sigma}\Theta^2\right) -{\rm tr}\left( \Theta \right) +\lambda_n\left\|\Theta\right\|_{1,\text{off}},
\end{equation}
where $\left\|A\right\|_{1,\text{off}}=\sum_{j\neq k}\left| A_{jk}\right|$ and ${\rm tr}\left( A\right)$ is the trace of the matrix A. We show that plugging-in $\hat{\Sigma}$ to the penalized D-trace loss also can give a consistent estimator $\hat{\Theta}$ of $\Theta$. Numerical analysis (see Table \ref{3}) shows that, compared to $\check{\Sigma}$, using $\hat{\Sigma}$ can significantly accelerate the optimization of the penalized D-trace loss.

The above optimization problem (\ref{eaa}) can be solved by an alternating direction method of multipliers \citep{Zhang2014}. In this paper, we use a more efficient algorithm developed in \cite{Wang2020} to calculate $\hat{\Theta}$. The tuning parameter $\lambda_n$ is selected by minimizing the following approximate Bayesian information criterion (BIC) \citep{Zhao2014},
$$\left\| {\frac{1}{2}(\hat{\Theta}\hat{\Sigma} + \hat{\Sigma} \hat{\Theta}) - I_p}\right\|_{\text{F}} + \left( {\|\hat{\Theta}\|}_{0}\log n\right) /n,$$
where $\| A\|_{\text{F}} $ is the Frobenius norm and $\| A\|_0 $ is the number of nonzero elements for matrix A.

\section{Theoretical properties}
\subsection{Notation}
We establish the theoretical property in the high dimensional setting. We first prove that the covariance matrix estimator $\hat{\Sigma}$ is a consistent estimator for $\Sigma$. The convergence rate for $\hat{\Sigma}$ is similar to that of the sample covariance matrix for random variables with polynomial tail probabilities. Then, under the same irrepresentability condition in \cite{Zhang2014}, we derive the edge recovery property and consistency for the PLNet estimator $\hat{\Theta}$. Meanwhile, we claim that $\check{\Sigma}$ and $\check{\Theta}$ (see equation \ref{ea}) has the same properties as $\hat{\Sigma}$ and $\hat{\Theta}$. All the detail proofs are shown in Supplementary Materials. 

Let $G=\{(i,j)|\Theta_{ij}\neq0\}$ be positions of non-zero elements in $\Theta$, $G^c$ be the complement set of $G$, $d$ is the maximum node degree in $\Theta$ and $s$ is the total number of nonzero elements of $\Theta$, $\theta_{\text{min}}=\min_{\left( i,j\right) \in S}\left|\Theta_{ij}\right|$ is the minimal absolute value of all nonzero elements of $\Theta$. For a vector $a=\left( a_1,\ldots,a_n\right) $, let $\|a\|_1=\sum_{i=1}^n \left| a_i\right| $ and $\|a\|_2=\left(\sum_{i=1}^n a_i^2\right)^{1/2}$ be the $l_1$- and $l_2$-norm of $a$. For a matrix $A$, let $\|A\|_{\infty}={\rm max}_{i,j}\left| A_{ij}\right| $ be the element-wise $l_\infty$-norm, $\|A\|_1=\sum_{i,j}\left| A_{ij}\right| $ be the $l_1$-norm, $\|A\|_{1,\infty}={\rm max}_i(\sum_j\left| A_{ij}\right| )$ be the $l_{1,\infty}$-norm, $\|A\|_F=\left( \sum_{i,j}\left| A_{ij}\right|^2\right) ^{1/2}$ be the Frobenius norm and $\|A\|_2={\rm max}_{\|v\|_2=1}\|Av\|_2$ be the operator norm. Further denote $\|A\|_{1,\text{off}}=\sum_{i\neq j}\left| A_{ij}\right| $. Let $\lambda_{\text{max}}(A)$ and $\lambda_{\text{min}}(A)$ be the largest and smallest eigenvalues of a symmetric matrix $A$. For any subset T of $\{\left( i,j\right) | i,j= 1,...,p \}$, let $\text{vec}\left( A\right) _T$ be the subvector of $\text{vec}(A)$ indexed by T. Let $\otimes$ be the Kronecker product, we define $\Gamma=\Gamma(\Sigma)=\left(\Sigma\otimes I+I\otimes\Sigma\right)/2,\ \hat{\Gamma}=\Gamma\left(\hat{\Sigma}\right)$. Let $(i,j), (k,l)\in \left\lbrace (i,j)|i,j =1,...,p\right\rbrace $. For any matrix $A_{p\times p}$, the row $(i-1)p+j$ and column $(k-1)p+l$ of the matrix $\Gamma\left( A\right)$ is $\Gamma\left( A\right)_{(i,j),(k,l)}=\left( A_{ik}h_{jl}+A_{jl}h_{ik}\right)/2 $, where $h_{ij}=1$ if $i=j$ and $h_{ij}=0$ if $i\neq j$. For two subset $T_1$ and $T_2$ of $\{\left( i,j\right) | i,j= 1,...,p \}$, we define $\Gamma\left( A\right)_{T_1,T_2}$ be the submatrix of $\Gamma\left( A\right)$ whose rows and columns indexed by $T_1$ and $T_2$, respectively. Other notations are as follows, $\gamma=1-\max_{(i,j)\in G^c}\left\|\Gamma_{(i,j),G}(\Gamma_{G,G})^{-1}\right\|_1,$
$k_{\Gamma}=\left\|\Gamma_{G,G}^{-1}\right\|_{1,\infty},k_{\Sigma}=\left\|\Sigma\right\|_{1,\infty}.$

\subsection{Irrepresentability condition and rate of convergence}
We first present the necessary irrepresentability condition for establishing the rate of convergence for the estimator in PLNet. This irrepresentability condition is from the D-trace method in \cite{Zhang2014} and is equvalent to $\gamma>0$ .
\begin{condition}[Irrepresentability condition]\label{con}
	$\max_{(i,j)\in G^c}\left\|\Gamma_{(i,j),G}(\Gamma_{G,G})^{-1}\right\|_1<1$.~
\end{condition}
\noindent We also need a boundedness condition for the true parameters in the PLN model (\ref{equ1}).
\begin{condition}[Boundedness condition]\label{conC}
	$\max_{1\leq i\leq n,1\leq j,k\leq p}\left\lbrace S_i, \left| \mu_j\right| , \left| \sigma_{jk}\right|  \right\rbrace  \leq C$ {\rm for some positive constant} $C>0$.~
\end{condition}

\noindent From the boundedness condition \ref{conC}, we establish the convergence rate for the covariance matrix estimator $\hat{\Sigma}$ and $\check{\Sigma}$.

\begin{theorem}[Rate of convergence for the covariance matrix estimator]\label{thm:rate}
	Under the boundedness condition \ref{conC}, for any positive integer $m$ and $0<\epsilon<6$, there exist constants $C_1$ and $C_2$ depending only on $m$, such that $pr\left(\left\|\hat{\Sigma}-\Sigma\right\|_{\infty}>\epsilon\right)<p^2/\left( C_1 n^m\epsilon^{2m}\right) $ and $pr\left(\left\|\check{\Sigma}-\Sigma\right\|_{\infty}>\epsilon\right)<p^2/\left( C_2 n^m\epsilon^{2m}\right) $.
\end{theorem}

\noindent Then, plugging in $\hat{\Sigma}$ to the lasso penalized D-trace loss, we get a consistent estimator $\hat{\Theta}$ that converges to $\Theta$ in several matrix norms.

\noindent
\begin{theorem}[Rate of convergence]\label{thm:recovery}
	Under the irrepresentability condition \ref{con} and the boundedness condition \ref{conC}, for any positive integer $m$, there exists a constant $C_1$ that only depends on $m$, for some $\eta >2$, choosing
	\begin{eqnarray*}
		n>&C_1^{-1/m}p^{\eta/m}{\rm max}\Bigg[12dk_{\Gamma},~12\gamma^{-1}(k_{\Sigma}k_{\Gamma}^2+k_{\Gamma}),~\left\lbrace 12\gamma^{-1}\left(k_{\Sigma}k_{\Gamma}^3+k_{\Gamma}^2\right)+5dk_{\Gamma}^2\right\rbrace \theta_{\min}^{-1},\\
		&{\rm min}\left\lbrace s^{1/2},d+1\right\rbrace \left\lbrace 12\gamma^{-1}\left(k_{\Sigma}k_{\Gamma}^3+k_{\Gamma}^2\right)+5dk_{\Gamma}^2\right\rbrace \lambda_{\min}^{-1}(\Theta),1/5\Bigg] ^2,
	\end{eqnarray*}
	and $$\lambda=12\gamma^{-1}\left(k_{\Sigma}k_{\Gamma}^2+k_{\Gamma}\right)C_1^{-1/(2m)}p^{\eta/(2m)}n^{-1/2}$$
	then with probability $1-p^{2-\eta}$,
	\begin{equation*}
	\begin{aligned}
	&\left\|\hat{\Theta}-\Theta\right\|_{\infty}\leq\left(12\gamma^{-1}\left(k_{\Sigma}k_{\Gamma}^3+k_{\Gamma}^2\right)+5dk_{\Gamma}^2\right)C_1^{-1/(2m)}p^{\eta/(2m)}n^{-1/2},~\\
	&\left\|\hat{\Theta}-\Theta\right\|_F\leq s^{1/2}\left(12\gamma^{-1}\left(k_{\Sigma}k_{\Gamma}^3+k_{\Gamma}^2\right)+5dk_{\Gamma}^2\right)C_1^{-1/(2m)}p^{\eta/(2m)}n^{-1/2},~\\
	&\left\|\hat{\Theta}-\Theta\right\|_2\leq {\rm min}\left\{s^{1/2},d+1\right\}\left(12\gamma^{-1}\left(k_{\Sigma}k_{\Gamma}^3+k_{\Gamma}^2\right)+5dk_{\Gamma}^2\right)C_1^{-1/(2m)}p^{\eta/(2m)}n^{-1/2},~		
	\end{aligned}
	\end{equation*}
\end{theorem}

\noindent Meanwhile, with a high probability, the sign of the sparse precision matrix $\Theta$ can be recovered by $\hat{\Theta}$ and we have the following theorem about the sign consistency of $\hat{\Theta}$.

\begin{theorem}[Sign consistency]\label{thm:sign}
	Under all the conditions in Theorem \ref{thm:recovery}, for some $\eta >2$, choosing the same $n$ and $\lambda$ in Theorem \ref{thm:recovery}, then with probability $1-p^{2-\eta}$,  $\hat{\Theta}$ recovers all zeros and nonzeros in $\Theta$.
\end{theorem}
\noindent These properties in Theorem \ref{thm:rate}, \ref{thm:recovery} and \ref{thm:sign} still hold for $\check{\Sigma}$ and $\check{\Theta}$.

\begin{theorem}[Rate of convergence and sign consistency for $\check{\Theta}$]\label{thm:both}
	Under all the conditions in Theorem \ref{thm:recovery}, the properties in Theorem \ref{thm:recovery} and Theorem \ref{thm:sign} also hold for $\check{\Theta}$ with another constant $C_2$. 
\end{theorem}
The rate-of-convergence and sign consistency results in Theorem \ref{thm:recovery}, \ref{thm:sign} and \ref{thm:both} are closely related to the polynomial tail situation of \cite{Zhang2014}. The boundedness condition \ref{conC} and the requirement $\epsilon<6$ are assumed because the log-transformation is not Lipschitz near zero. Ignoring the complicated constants in theorem, for any $\eta >2 $ and any positive integer $m$, if we have $p < o(n^{m/\eta})$, or in other words, if $p$ tends to infinity not too fast, $\hat{\Theta}$ and $\check{\Theta}$ are consistent estimators of $\Theta$. Especially, the rate of convergence for $\hat{\Theta}$ and $\check{\Theta}$ is $O\left( \left( p^{\eta/m}/n\right)^{1/2} \right) $ under $l_{\infty}$-norm.

\section{Simulation studies}

\subsection{Simulation settings}

We conduct simulations to evaluate the performance of PLNet and compare with the available network inference methods including VPLN \citep{Chiquet2019} and glasso \citep{Friedman2008}. Both PLNet and VPLN are designed to estimate the precision matrix for count data in the PLN model. The glasso algorithm is a classical approach for continuous data in GGM and we apply glasso to the logarithmic transformation of the normalized data, which is defined as $\tilde{Y}_{ij} = (Y_{ij}+1)/\sum_{j = 1}^{p}Y_{ij},~i = 1,2,\cdots,n$, where $Y=\left[ Y_{ij}\right]_{1\leq i \leq n,1\leq j\leq p} $ is the observed count matrix with $n$ rows (cells) and $p$ columns (genes). We add 1 to all counts before taking normalization since there are many 0's in the count data. In all simulations, we estimate the library size $S_i$ for all methods by total sum scaling, which is a classical normalization method for scRNA-seq and is defined as the sum of counts within each cell.

We simulate count data from the PLN model with different choices of library sizes, the mean vectors and the precision matrices. The library sizes are generated from a log-normal distribution $ \text{N}\left( \log10,{\sigma_0}^2\right) $, with ${\sigma_0} = {0.1}$ or $ {0.3}$ representing low and high variations of library sizes across samples, respectively. The mean vector $\bmu$ is set as $\bmu=\left( -1.8,\dots,-1.8\right)^T$ or $\bmu = \left( -2.8,\dots,-2.8\right)^T$, where the former corresponds to a low-dropout scenario (about 10 percent of the counts are zero) and the latter to a high-dropout scenario (about 30 percent of the counts are zero). We consider the following four graph structures:
\begin{itemize}
	\item[1.] Banded Graph: Pairs $(i,j)$ of nodes are connected if $|i-j|\leq 2,~i\neq j$. All nonzero edges are set as $0.3$.
	
	\item[2.] Random Graph: Pairs of nodes are connected with probability $0.1$. The nonzero edges are set as $0.3$ with probability $0.8$ and as $-0.3$ with probability $0.2$.
	
	\item[3.] Scale-free Graph: The Barabasi-Albert model \citep{Barabasi1999} is used to generate a scale-free graph with power $1$. The nonzero edges are set as $0.3$.
	
	\item[4.] Blocked Graph: The nodes are divided into $5$ blocks of equal sizes. Pairs of nodes in the same block are connected with probability $0.1$ and the nonzero edges are set as $0.3$. Blocks are separated and different blocks have no edge connection.
\end{itemize}

The diagonal elements of the precision matrices are all first set as 1. If a precision matrix is not positive definite, a positive number is added in the diagonal elements of the precision matrix to guarantee positive definiteness. For all simulation data, we set the sample size $n=2000$ and consider different gene numbers $p=100,200 \text{ and } 300$. For each combination of model settings, we independently repeat simulations 100 times.

\subsection{Performance comparison}

Table \ref{1} shows the area under precision and recall curve (AUPR) of each estimator. AUPRs are calculated by varying the tuning parameters (i.e. the penalty parameters of the three estimators). As expected, the AUPR decreases as the number of genes increases. AUPRs in the high-dropout cases are generally smaller than in the low-dropout cases. Scenarios with a high variation of the library size also generally have smaller AUPRs than scenarios with a low variation. PLNet is the most robust estimation among these three estimators and outperforms VPLN and glasso in almost all simulation settings in AUPR, especially for the settings with high dropouts or with high variations of the library size. For example, for $p=100$, PLNet achieves an AUPR of 0.95 for the banded graph under the scenario of the high dropout rate and high variation, while VPLN and glasso only have AUPRs of 0.42 and 0.13, respectively. The results of the area under the Receiver Operating Characteristic curve (AUC) are similar and shown in Supplementary Material.
\begin{table}
	\footnotesize
	\centering
	\caption{Comparisons of PLNet with VPLN and glasso in terms of the area under precision and recall curve (AUPR) on simulation results. The results are averages over 100 replicates with standard deviations in brackets}
	\footnotesize
	\begin{tabular}{cccccccc}
		Library size&&\multicolumn{2}{c}{$p=100$}&\multicolumn{2}{c}{$p=200$}&\multicolumn{2}{c}{$p=300$}\\
		variation&Dropout&Low&High&Low&High&Low&High\\
		&&\multicolumn{6}{c}{Banded graph}\\
		&PLNet&\textbf{0.98 (0.01)} & \textbf{0.95 (0.01)} & \textbf{0.95 (0.01)} & \textbf{0.92 (0.01)} & \textbf{0.92 (0.01)} & \textbf{0.88 (0.01)} \\
		Low&VPLN&0.9 (0.03) & 0.46 (0.11) & 0.9 (0.01) & 0.52 (0.02) & 0.91 (0.01) & 0.54 (0.05) \\
		&glasso&0.85 (0.01) & 0.15 (0.02) & 0.9 (0.01) & 0.38 (0.03) & 0.92 (0.01) & 0.51 (0.03) \\
		&PLNet&\textbf{0.98 (0.01)} & \textbf{0.95 (0.01)} & \textbf{0.95 (0.01)} & \textbf{0.92 (0.01)} & \textbf{0.92 (0.01)} & \textbf{0.87 (0.01)} \\
		High&VPLN&0.81 (0.15) & 0.42 (0.15) & 0.82 (0.08) & 0.47 (0.13) & 0.66 (0.03) & 0.48 (0.06) \\
		&glasso&0.65 (0.03) & 0.13 (0.04) & 0.73 (0.01) & 0.32 (0.05) & 0.76 (0.01) & 0.42 (0.04) \\
		&&\multicolumn{6}{c}{Random graph}\\
		&PLNet&\textbf{0.83 (0.03)} & \textbf{0.66 (0.05)} & \textbf{0.65 (0.02)} & \textbf{0.41 (0.03)} & \textbf{0.52 (0.02)} & \textbf{0.29 (0.02)} \\
		Low&VPLN&0.59 (0.03) & 0.23 (0.04) & 0.45 (0.03) & 0.18 (0.02) & 0.38 (0.02) & 0.14 (0.01) \\
		&glasso&0.56 (0.03) & 0.23 (0.02) & 0.47 (0.02) & 0.19 (0.01) & 0.40 (0.02) & 0.15 (0.01) \\
		&PLNet&\textbf{0.82 (0.03)} & \textbf{0.64 (0.05)} & \textbf{0.64 (0.03)} & \textbf{0.38 (0.03)} & \textbf{0.50 (0.02)} & \textbf{0.27 (0.02)} \\
		High&VPLN&0.50 (0.05) & 0.12 (0.02) & 0.32 (0.03) & 0.11 (0.01) & 0.23 (0.03) & 0.11 (0.01) \\
		&glasso&0.38 (0.03) & 0.16 (0.02) & 0.28 (0.02) & 0.14 (0.01) & 0.22 (0.01) & 0.13 (0.01) \\
		&&\multicolumn{6}{c}{Scale-free Graph}\\
		&PLNet&\textbf{0.79 (0.14)} & \textbf{0.61 (0.14)} & \textbf{0.58 (0.13)} & \textbf{0.38 (0.14)} & \textbf{0.46 (0.16)} & \textbf{0.27 (0.11)} \\
		Low&VPLN&0.64 (0.18) & 0.28 (0.06) & 0.51 (0.14) & 0.21 (0.04) & 0.45 (0.14) & 0.17 (0.04) \\
		&glasso&0.55 (0.15) & 0.31 (0.07) & 0.45 (0.14) & 0.23 (0.06) & 0.42 (0.14) & 0.19 (0.05) \\
		&PLNet&\textbf{0.80 (0.14)} & \textbf{0.59 (0.13)} & \textbf{0.53 (0.16)} & \textbf{0.34 (0.14)} & \textbf{0.50 (0.17)} & \textbf{0.28 (0.11)} \\
		High&VPLN&0.55 (0.17) & 0.16 (0.06) & 0.36 (0.13) & 0.10 (0.05) & 0.32 (0.12) & 0.08 (0.04) \\
		&glasso&0.46 (0.09) & 0.26 (0.06) & 0.33 (0.08) & 0.20 (0.04) & 0.30 (0.09) & 0.14 (0.03) \\
		&&\multicolumn{6}{c}{Blocked graph}\\
		&PLNet&\textbf{0.70 (0.03)} & \textbf{0.58 (0.04)} & \textbf{0.57 (0.03)} & \textbf{0.41 (0.03)} & \textbf{0.48 (0.02)} & \textbf{0.32 (0.03)} \\
		Low&VPLN&0.68 (0.03) & 0.33 (0.06) & 0.53 (0.05) & 0.26 (0.04) & 0.46 (0.02) & 0.22 (0.04) \\
		&glasso&0.56 (0.03) & 0.28 (0.02) & 0.52 (0.02) & 0.26 (0.02) & 0.46 (0.02) & 0.23 (0.02) \\
		&PLNet&\textbf{0.71 (0.03)} & \textbf{0.57 (0.04)} & \textbf{0.56 (0.02)} & \textbf{0.41 (0.03)} & \textbf{0.47 (0.02)} & \textbf{0.30 (0.03)} \\
		High&VPLN&0.59 (0.04) & 0.19 (0.03) & 0.40 (0.04) & 0.16 (0.02) & 0.32 (0.03) & 0.14 (0.01) \\
		&glasso&0.43 (0.03) & 0.21 (0.02) & 0.35 (0.03) & 0.19 (0.01) & 0.3 (0.02) & 0.17 (0.01)	
	\end{tabular}
	\label{1}
\end{table}

We also compare the true positive rates (TPR), the true discovery rates (TDR), and the Frobenius risks of the three estimators with the tuning parameters selected based on the BIC. The Frobenius risk is defined as the Frobenius norm of the difference between the true and estimated precision matrices. Table \ref{2} summarizes these results for the random graphs. In the low dropout case, PLNet achieves acceptable TPR and much higher TDR than those of the other two methods in most cases. In the high dropout case, TPR and TDR of PLNet are higher than that of the other two methods in most cases. The Frobenius risks of PLNet are also smaller than those of the other two methods in many cases. The results of the other three types of graphs are similar to the random graph and are shown in Supplementary Material. 
\begin{table}
	\footnotesize
	\centering
	\caption{Comparisons of PLNet with VPLN and glasso in terms of the true positive rate (TPR), the true discovery rate (TDR), and the Frobenius risk for random graph. The tuning parameters of the three methods are selected by BIC criterion. The results are averages over 100 replicates with standard deviations in brackets}
	\begin{tabular}{cccccccc}
		&&\multicolumn{2}{c}{$p=100$}&\multicolumn{2}{c}{$p=200$}&\multicolumn{2}{c}{$p=300$}\\
		&Dropout&Low&High&Low&High&Low&High\\
		&&\multicolumn{6}{c}{Low library size variation}\\
		&PLNet&0.87 (0.05)&\textbf{0.63 (0.07)} &0.25 (0.07)&\textbf{0.12 (0.03)} & 0.05 (0.02) & \textbf{0.02 (0.01)}\\
		TPR&VPLN&\textbf{0.93 (0.15)} & 0.12 (0.11) & \textbf{0.50 (0.31)} & 0.02 (0.01) & \textbf{0.52 (0.15)} & 0.01 (0.01) \\
		&glasso&0.65 (0.09) & 0.17 (0.04) & 0.18 (0.04) & 0.04 (0.01) & 0.06 (0.02) & 0.01 (0.01) \\
		&PLNet&\textbf{0.66 (0.05)} & \textbf{0.62 (0.05)} & \textbf{0.84 (0.03)} & \textbf{0.76 (0.05)} & \textbf{0.87 (0.04)} & \textbf{0.79 (0.07)} \\
		TDR&VPLN&0.36 (0.09) & 0.38 (0.08) & 0.45 (0.10) & 0.45 (0.09) & 0.40 (0.06) & 0.41 (0.11) \\
		&glasso&0.51 (0.04) & 0.39 (0.05) & 0.61 (0.04) & 0.47 (0.05) & 0.64 (0.05) & 0.44 (0.05) \\
		&PLNet&\textbf{7.57 (0.56)} & \textbf{8.97 (0.37)} & 22.31 (0.57) & 23.80 (0.34) & 36.41 (0.67) & 38.06 (0.37) \\
		Frobenius risk&VPLN&9.01 (2.76) & 17.08 (1.65) & 25.15 (5.27) & 35.55 (1.36) & 35.92 (3.60) & 53.79 (1.77) \\
		&glasso&8.43 (0.28) & 12.04 (0.29) & \textbf{21.90 (0.21)} & \textbf{19.89 (0.11)} & \textbf{36.10 (0.24)} & \textbf{28.52 (0.01)} \\
		&&\multicolumn{6}{c}{High library size variation}\\
		&PLNet&0.86 (0.05)&\textbf{0.60 (0.08)} & 0.24 (0.05)&\textbf{0.10 (0.02)} & 0.05 (0.02) & \textbf{0.08 (0.01)}\\
		TPR&VPLN&\textbf{0.89 (0.19)} & 0.15 (0.13) & \textbf{0.59 (0.11)} & 0.03 (0.04) & \textbf{0.36 (0.11)} & 0.07 (0.01) \\
		&glasso&0.44 (0.06) & 0.34 (0.07) & 0.12 (0.03) & 0.04 (0.01) & 0.04 (0.01) & 0.01 (0.01) \\
		&PLNet&\textbf{0.66 (0.04)} & \textbf{0.64 (0.05)} & \textbf{0.85 (0.03)} & \textbf{0.75 (0.06)} & \textbf{0.85 (0.04)} & \textbf{0.75 (0.07)} \\
		TDR&VPLN&0.31 (0.05) & 0.22 (0.08) & 0.27 (0.03) & 0.31 (0.08) & 0.24 (0.04) & 0.25 (0.18) \\
		&glasso&0.38 (0.04) & 0.09 (0.04) & 0.43 (0.03) & 0.33 (0.04) & 0.44 (0.04) & 0.33 (0.05) \\
		&PLNet&\textbf{7.61 (0.47)} & \textbf{9.16 (0.36)} & \textbf{22.31 (0.43)} & 23.88 (0.34) & \textbf{36.34 (0.51)} & 38.18 (0.39) \\
		Frobenius risk&VPLN&9.39 (2.23) & 16.91 (2.23) & 22.50 (1.83) & 36.00 (2.30) & 37.65 (4.59) & 50.89 (7.42) \\
		&glasso&9.19 (0.17) & 14.04 (0.96) & 22.44 (0.22) & \textbf{19.00 (0.04)} & 36.76 (0.22) & \textbf{30.11 (0.19)} 	
	\end{tabular}
	\label{2}
\end{table}

To further demonstrate the performance of PLNet, we visualize the mean networks predicted by the three methods for the banded graph with $p=100$ over the 100 simulations (Fig. \ref{fig:Banded-graph frequency}). More specifically, we calculate the relative frequency $F_{ij}$ that an algorithm reports edges between nodes $i$ and $j$ ($i,j=1,\cdots,p$) over the 100 simulations. For positions $(i,j)$ with $\Theta_{ij}\neq 0$, $F_{ij}$ is the proportion that an algorithm correctly recovers the edge in the 100 simulations, while for positions $(i,j)$ with $\Theta_{ij} = 0$, $F_{ij}$ is the the proportion that an algorithm falsely predicts edges between nodes $i$ and $j$ in the 100 simulations. We plot the relative frequency matrices of the three methods in Fig. \ref{fig:Banded-graph frequency}. The frequencies are represented by colors from red to blue with false edges colored in red and true predictions in blue. We also plot the true network matrix for reference. We clearly see that PLNet is able to detect more true positives while having less false positives than the other methods, especially for the high dropout case. 
\begin{figure}
	\centering
	\includegraphics[width=6in]{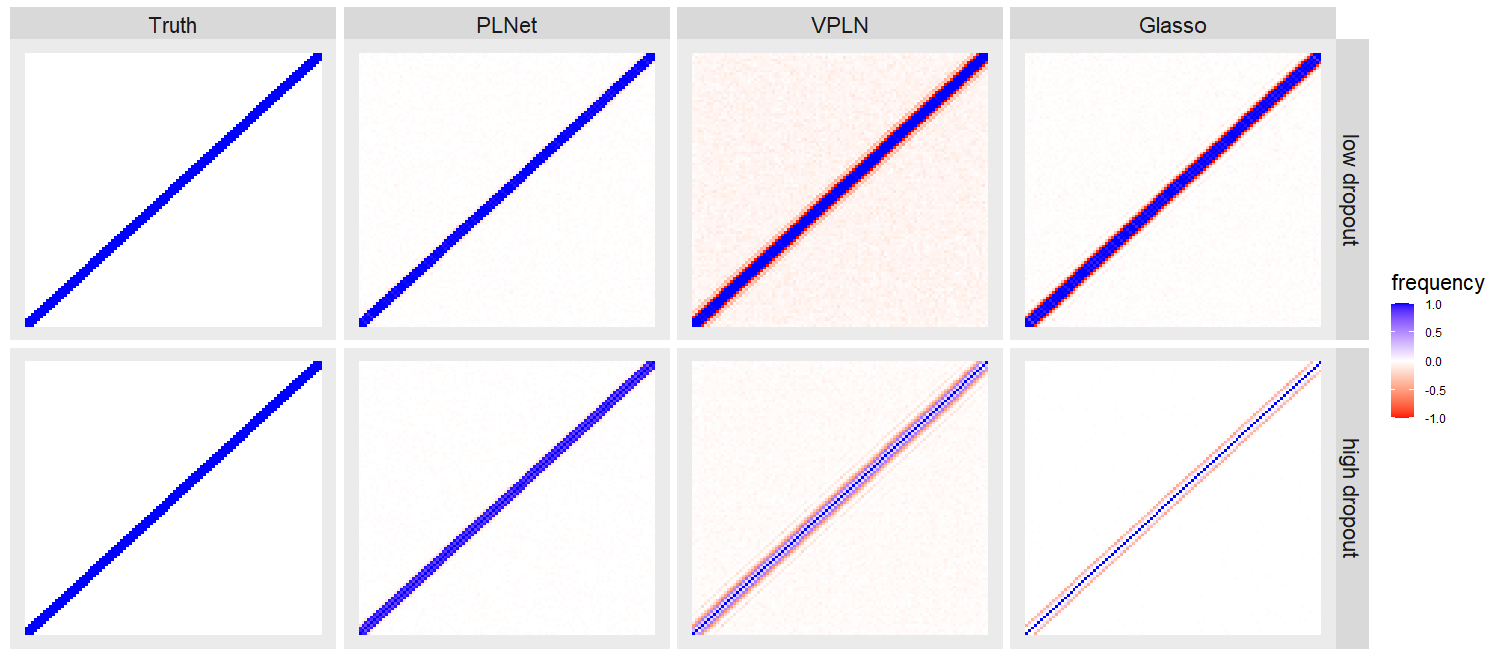}
	\caption{The mean networks predicted by PLNet with VPLN and glasso for the banded graph with 100 nodes. False edges are colored in red and true edges are in blue.}
	\label{fig:Banded-graph frequency}
\end{figure}

Table \ref{3} shows the mean computational time of the three algorithms. We also include the computational time of $\check{\Theta}$ (plugging-in $\check{\Sigma}$ in the lasso penalized D-trace loss) and denote it as PLNet* in the table. The glasso method is computationally the most efficient since its optimization problem is much simpler than that of PLNet and VPLN. PLNet is computationally more efficient than VPLN, sometimes by a very large amount. Interestingly, we observe that VPLN generally takes much more time for the high dropout cases than the low dropout cases. In comparison, the computational efficiency of PLNet is roughly the same for the low and the high dropout cases. VPLN is computationally less efficient because VPLN involves a series of glasso optimizations from the variational approximation. PLNet is computational more expensive than glasso because it needs to first find a projection of the estimator $\tilde{\Sigma}$ in the semi-definite matrix space. Finally, PLNet is generally more efficient than PLNet*. In extreme cases, the computational time of PLNet is only about $22\%$ of PLNet*. 
\begin{table}
	\footnotesize
	\centering
	\caption{Comparison of PLNet with VPLN and glasso in terms of CPU time (minute).``PLNet*" is the PLNet-based estimator $\check{\Theta}$. The results are averages over 200 replicates including cases of low and high library size variation. Numbers in brackets are standard deviations}
	\begin{tabular}{ccccccccc}
		&\multicolumn{2}{c}{$p=100$}&&\multicolumn{2}{c}{$p=200$}&&\multicolumn{2}{c}{$p=300$}\\
		Dropout&Low&High&&Low&High&&Low&High\\
		&\multicolumn{8}{c}{Banded graph}\\
		PLNet&4.20 (0.42)&4.32 (0.36)&&10.08 (0.90)&34.50 (2.10)&&20.10 (2.22)&118.92 (6.00)\\
		PLNet*&5.28 (0.57)&6.51 (0.54)&&13.85 (1.24)&56.77 (3.46)&&32.63 (3.60)&216.72 (10.93)\\
		VPLN&3.78 (0.18)&5.58 (1.26)&&14.04 (1.38)&151.80 (60.18)&&29.58 (3.54)&994.38 (270.60)\\
		glasso&0.06 (0.06)&0.06 (0.06)&&0.12 (0.06)&0.24 (0.06)&&0.36 (0.06)&0.96 (0.06)\\
		&\multicolumn{8}{c}{Random graph}\\
		PLNet&0.42 (0.06)&0.48 (0.06)&&3.12 (0.42)&3.36 (0.30)&&15.42 (1.50)&9.90 (0.66)\\
		PLNet*&0.77 (0.11)&1.00 (0.13)&&4.68 (0.63)&6.22 (0.56)&&19.63 (1.91)&20.86 (1.39)\\
		VPLN&4.26 (0.78)&7.20 (0.96)&&12.30 (1.62)&22.86 (13.44)&&31.26 (6.12)&206.22 (240.66)\\
		glasso&0.06 (0.06)&0.06 (0.06)&&0.12 (0.06)&0.24 (0.06)&&0.36 (0.06)&0.90 (0.06)\\
		&\multicolumn{8}{c}{Scale-free graph}\\
		PLNet&0.30 (0.06)&0.30 (0.06)&&1.26 (0.06)&1.32 (0.06)&&3.54 (0.54)&104.10 (4.98)\\
		PLNet*&1.31 (0.26)&1.23 (0.25)&&5.16 (0.25)&5.90 (0.18)&&15.36 (2.39)&467.19 (22.35)\\
		VPLN&5.88 (0.96)&8.10 (1.62)&&24.30 (3.96)&79.44 (34.50)&&49.68 (33.60)&468.48 (185.52)\\
		glasso&0.06 (0.06)&0.06 (0.06)&&0.12 (0.06)&0.18 (0.06)&&0.24 (0.06)&0.84 (0.24)\\
		&\multicolumn{8}{c}{Banded graph}\\
		PLNet&0.30 (0.06)&0.30 (0.06)&&1.80 (0.30)&1.80 (0.30)&&110.52 (4.98)&114.84 (5.64)\\
		PLNet*&0.93 (0.19)&0.93 (0.19)&&3.84 (0.64)&4.72 (0.79)&&189.99 (8.46)&241.16 (11.84)\\
		VPLN&4.50 (0.78)&8.22 (3.12)&&12.54 (1.56)&39.78 (38.16)&&40.74 (7.32)&218.04 (296.76)\\
		glasso&0.06 (0.06)&0.06 (0.06)&&0.12 (0.06)&0.18 (0.06)&&0.30 (0.06)&0.78 (0.06)
	\end{tabular}
	\label{3}
\end{table}
\section{Application to a scRNA-seq dataset}
We apply PLNet and VPLN to infer the gene regulatory network of CD14+ Monocytes profiled in \citet{Kang2018}. The single cells are profiled in two different conditions, IFN-$\beta$-treated and control. IFN-$\beta$ is a cytokine in the interferon family that influences the transcriptional profiles for many genes, especially that in the JAK/STAT pathway \citep{Mostafavi2016}. We focus on the IFN-$\beta$-treated cells (2147 cells) and use the top 200 highly variable genes that are used in \cite{Stuart2019} for network analysis.

We first compare the networks of PLNet and VPLN with the parameters tuned such that the network densities are around 5\%. Gene Ontology (GO) analysis \citep{Kuleshov2016} shows that the 200 genes mainly involve in 4 major biological processes, including ``Cytokine-mediated signaling pathway" (Module $M_1$), ``neutrophil mediated immunity" (Module $M_2$), ``cellular protein metabolic process" (Module $M_3$), and ``proteolysis" (Module $M_4$). Figure \ref{fig:heatmap1} shows the predicted networks of the genes in the 4 modules by PLNet and VPLN, where the colors represent the partial correlations between genes. The partial correlation given by PLNet between genes $i$ and $j$ is defined as $-\hat{\Theta}_{ij}/(\hat{\Theta}_{ii}\hat{\Theta}_{ij})^{1/2}$. The partial correlation given by VPLN is defined similarly. We clearly see that the network given by PLNet tend to have more connections within the modules than VPLN. To see this more clearly, for each module $M_k$, we calculate the ratio between within-module and between-module connections $R(M_k) = \Sigma_{i,j\in M_k}W_{ij}/\Sigma_{i\in M_k,~j\notin M_k} W_{ij}$, where the weights $W_{ij}$ are set as the partial correlation between nodes $i$ and $j$ (weighted within-between connection ratio) or are set as 1 and 0 depending on whether nodes $i$ and $j$ are connected (unweighted within-between connection ratio). The within-between connection ratios of PLNet are much larger than that of VPLN in most cases (Table \ref{tab4}). Similar results also hold for networks of other densities or the networks chosen by the BIC (See Supplementary Material). 
\begin{figure}
	\centering
	\subfigure[PLNet]{
		\begin{minipage}[t]{0.5\linewidth}
			\includegraphics[width=3in]{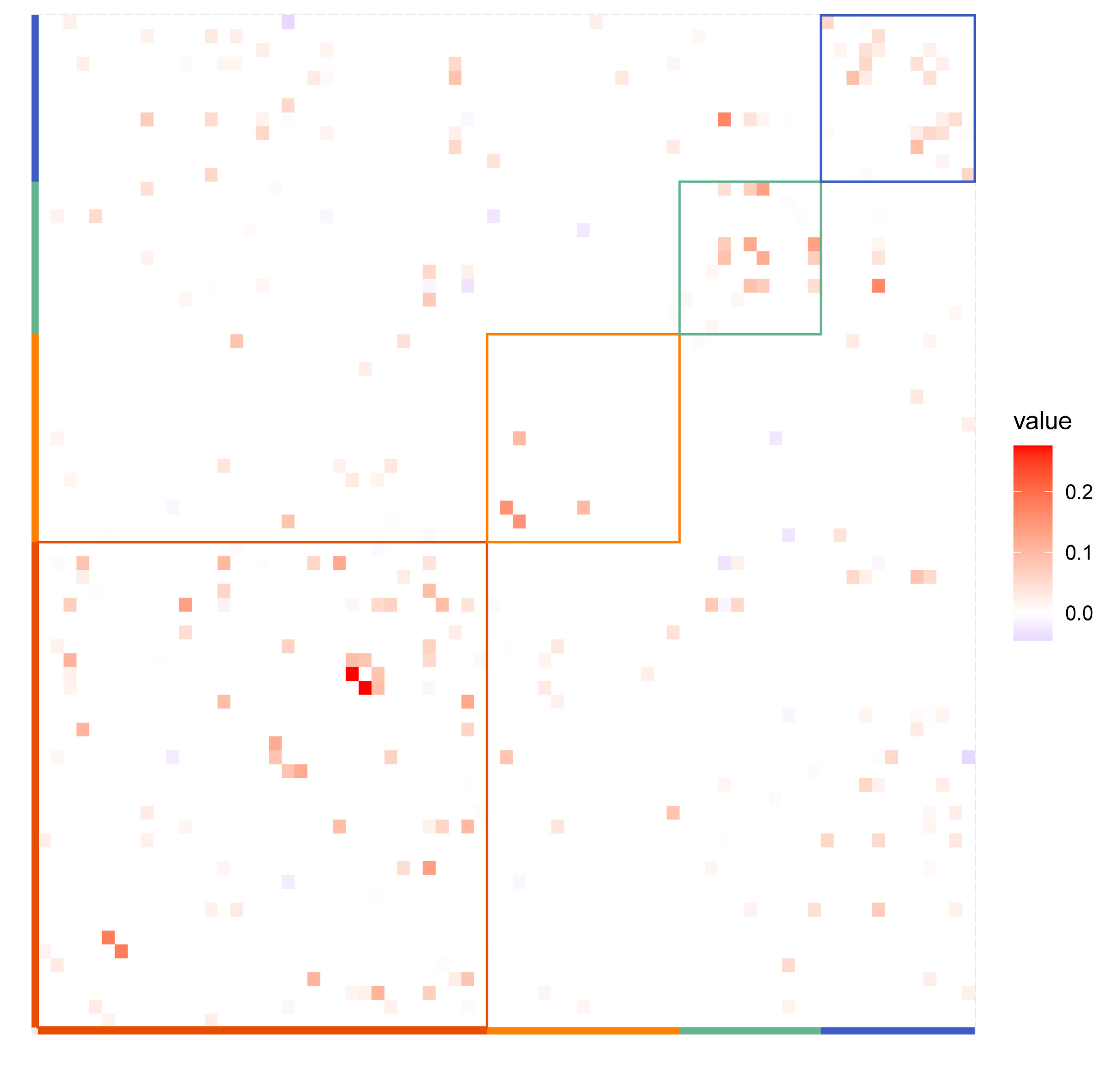}
		\end{minipage}%
	}%
	\subfigure[VPLN]{
		\begin{minipage}[t]{0.5\linewidth}
			\includegraphics[width=3in]{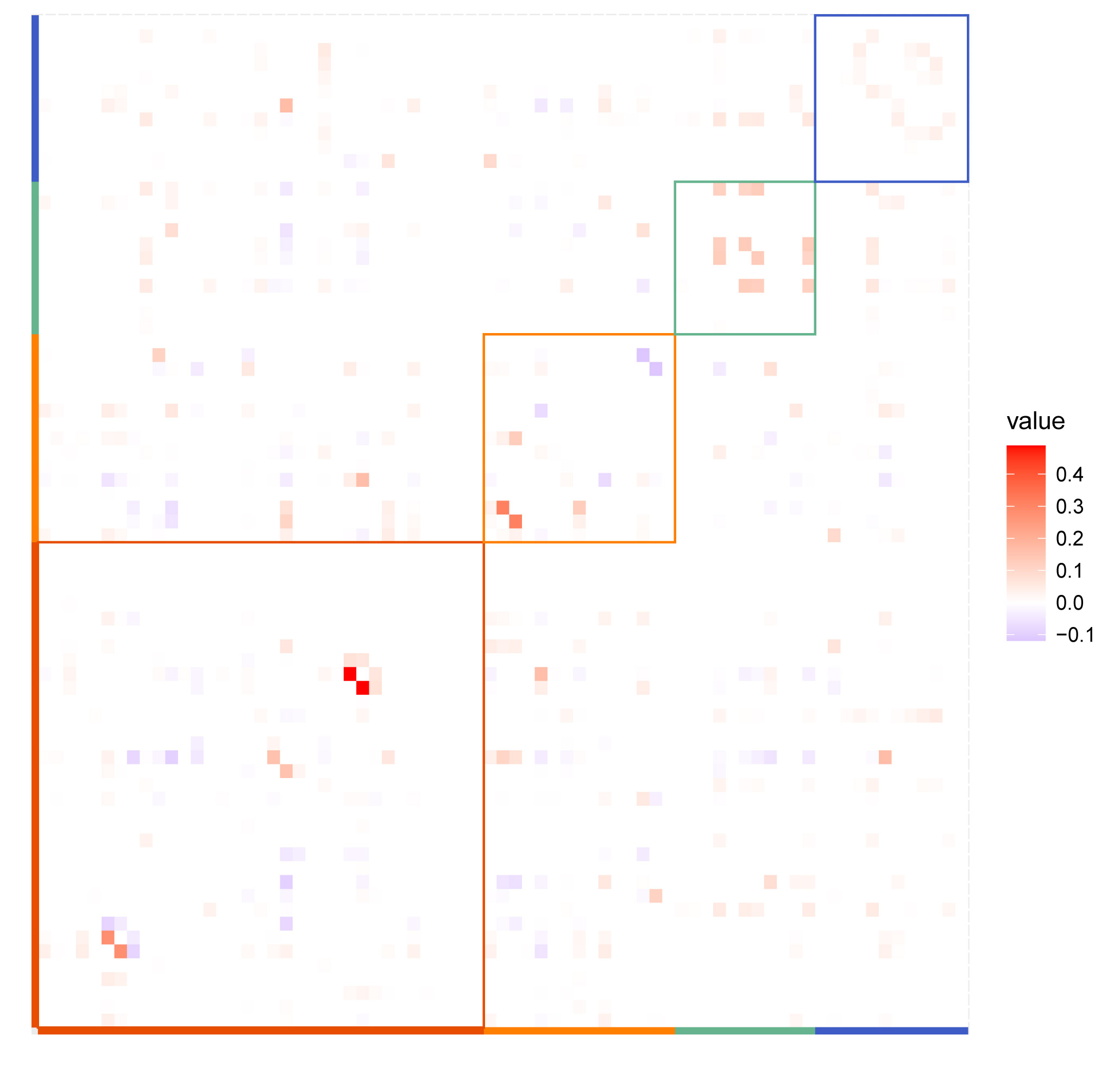}
		\end{minipage}%
	}%
	\caption{Heat maps of partial correlations between genes in the 4 GO modules given by PLNet (a) and VPLN (b). Red: cytokine-mediated signaling pathway (Module $M_1$); Yellow: neutrophil mediated immunity (Module $M_2$); Green: cellular protein metabolic process (Module $M_3$); Blue: proteolysis (Module $M_4$)}
	\label{fig:heatmap1}
\end{figure}
\begin{table}
	\footnotesize
	\centering
	\caption{The within-between connection ratios of the 4 major modules in the networks estimated by PLNet and VPLN tuned such that the network densities are around 5\%}
	\begin{tabular}{cccccc}
		Graph& Method &Module 1&Module 2&Module 3&Module 4\\
		Weighted &PLNet &\textbf{0.751} &\textbf{0.448}& \textbf{0.623}&\textbf{0.419}\\
		&VPLN &0.563 &0.401& 0.497&0.245\\
		Unweighted &PLNet &\textbf{0.597} &0.148& \textbf{0.429}&\textbf{0.393}\\
		&VPLN &0.467 &\textbf{0.216}& 0.171&0.22
	\end{tabular}
	\label{tab4}
\end{table}
We then compare the networks of PLNet and VPLN with the parameters tuned by the BIC. PLNet identifies 6 genes connecting to IFNB1, which encodes the IFN-$\beta$ protein, while VPLN does not find any genes connecting to IFNB1. The 6 genes are CCL13, CCL23, CXCL1, IL18, MT1G, and PRR16. Many of these edges connecting IFNB1 are probably true regulatory relationships. For example, CXCL1 and MT1G have been previously reported to be regulated by IFNB1 \citep{Jablonska2014,Hilpert2008}, while CCL13 and CCL23, two of Cys-Cys chemokine family members, are shown to be regulated by IFN-$\beta$ through the tumor necrosis factor-alpha (TNF-$\alpha$) \citep{Ozenci2000,Hornung2000}. Among the 200 genes, 14 genes (such as IFNB1 and MT1G) are only expressed in the IFN-$\beta$-treated cells. Presumably, these genes should be upregulated through a certain regulatory network upon IFN-$\beta$ stimulation and the inferred regulatory network should contain edges connecting these genes. In total, PLNet reports 238 edges including 41 edges connecting these 14 genes, and 11 of 14 genes have nonzero degrees. VPLN reports much more edges than PLNet. There are 428 edges, including only 3 edges involving the 14 genes, in the VPLN network, and only 2 of the 14 genes have nonzero degrees. We further plot the total degrees of these genes predicted by PLNet and VPLN as well as the expected degree of these genes in random networks at different network densities (Fig. \ref{fig:degreeplot}). The total degree of the 14 genes in the VPLN network is much smaller than in the PLNet network, and even smaller than in the random network. Among the 200 genes, 2 genes (MYC and KLF2) are transcription factors with available ChIP-seq data \citep{Rouillard2016,Lachmann2010,Encode2004}.  We find that among the genes detected to be the target of these two transcription factors by PLNet and VPLN, 83\% (6/7) and 73\% (8/11) of genes are supported by ChIP-seq experiments, respectively.      

\begin{figure}
	\centering
	\includegraphics[width=5in]{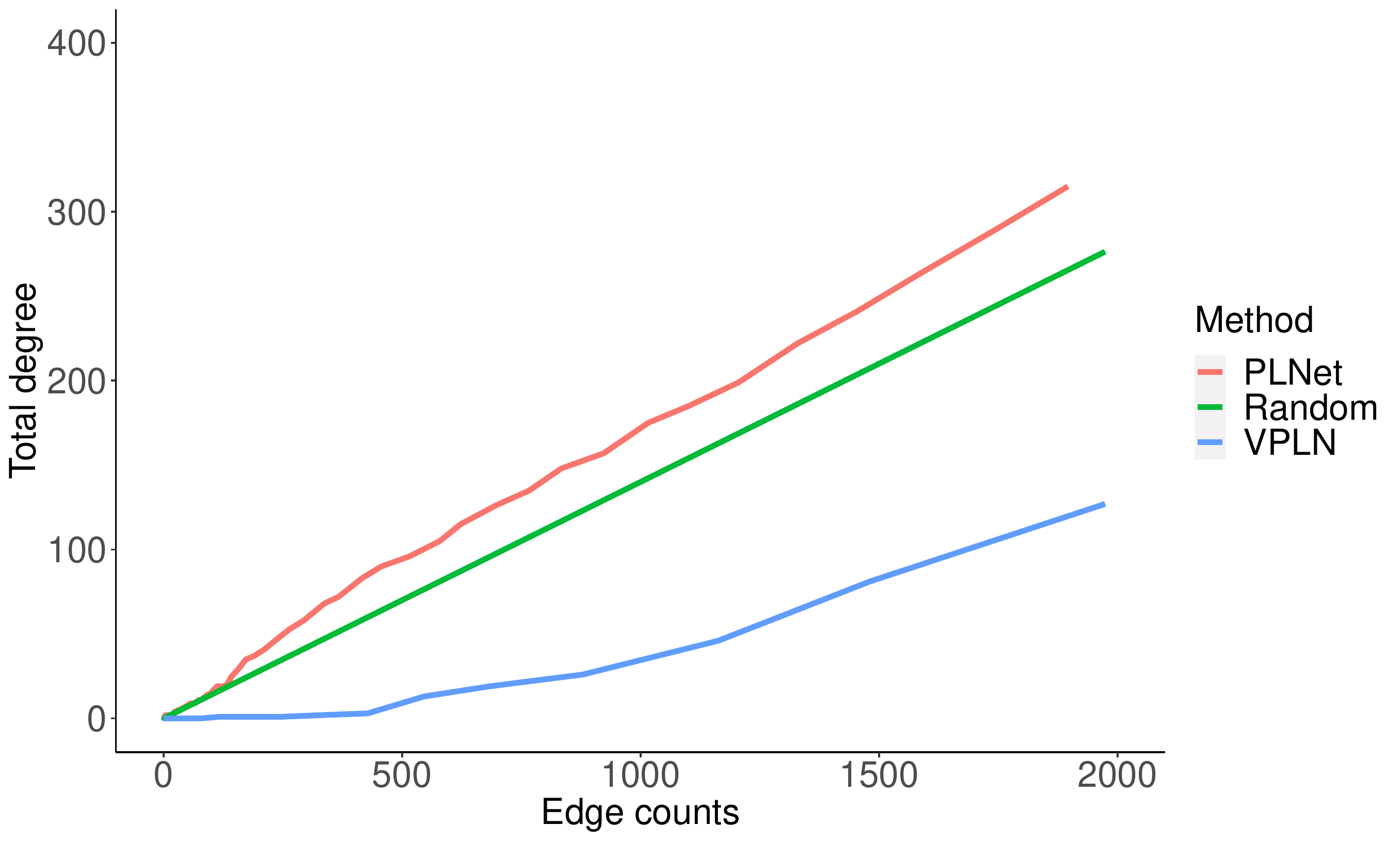}
	\caption{The total degrees of the 14 genes which are only expressed in the IFN-$\beta$-treated cells in the networks estimated by PLNet and VPLN at various network densities. Random refers to the expected degrees of the 14 genes for totally random networks with  various network densities.}
	\label{fig:degreeplot}
\end{figure}

\section{Discussion}
\label{sec4}

In this paper, we consider the PLN graphical model for count data. This model has an intuitive explanation for single-cell gene regulatory network analysis. To estimate the underlying precision matrix, we propose a two-step estimator, using the moment method to estimate the covariance matrix and then minimizing the penalized D-trace loss to estimate the precision matrix. The simplicity of this estimation procedure allows us to establish consistency theory for the proposed PLNet estimator even for the high dimensional setting. The numerical analysis also shows that the PLNet method outperforms available methods. 

The proposed method can be generalized in several ways. A straightforward generalization is to the differential network analysis based on our earlier work \citep{yuan2017differential} in single cells. Another generalization is gene regulatory network analysis of mixtures of cell populations. Different cell populations may have different gene regulatory networks and we could jointly model the mixture and infer the gene regulatory networks for all cell populations.

\section{Appendix}
\subsection{Technical proofs}
\subsubsection{Lammas and proofs}
We need two lemmas for the proofs of the theorems in the paper. Lamma \ref{dtrace} is the Lemma A1 (b) and (c) in D-trace method \citep{Zhang2014}, the proof of the lemma \ref{lem:3} is given in the Supplementary Material.

\begin{lemma}\label{dtrace}
	We define
	$$\breve{\Theta}=\argmin_{A=A^T}\frac{1}{2}{\rm tr}\left(\hat{\Sigma}A^2\right)-{\rm tr}\left(A\right)+\lambda\left\|A\right\|_{1,\text{off}}. $$
	Then the following hold:
	
	(a) $\text{vec}\left(\breve{\Theta}\right)_{G^c}=0$, if 
	$$\|\hat{\Sigma} - \Sigma \|_{\infty}<1/\left( 12dk_{\Gamma}\right) ,$$
	$$6\|\hat{\Sigma} - \Sigma \|_{\infty}\left(k_{\Sigma}k_{\Gamma}^2+k_{\Gamma}\right)\leq 0.5\gamma {\rm min}\left\{\lambda,1\right\};$$
	
	(b) assuming the conditions in part (a), we also have
	$$\left\|\breve{\Theta}-\Theta\right\|_{\infty}<\lambda k_{\Gamma}+\frac{5}{2}d\left(1+\lambda\right)\|\hat{\Sigma} - \Sigma \|_{\infty} k_{\Gamma}^2.$$
\end{lemma}

\begin{lemma}\label{lem:3}
	Let $\tilde{\Sigma}$ be the moment estimator in Equation (\ref{equ3}). 
	\begin{equation}\label{equ3}
	\tilde{\sigma}_{jk} = \begin{cases}
	\log⁡\left( n^{-1} \sum_{i=1}^{n}\left[ \left\lbrace Y_{ij}\left( Y_{ij}-1\right) \right\rbrace / S_i^2 \right] \right) - 2\log⁡\left( \tilde{\alpha}_j\right),\ &{\rm for}\  1\leq j=k \leq p, \\
	\log⁡\left[ n^{-1} \sum_{i=1}^{n}\left\lbrace \left( Y_{ij}Y_{ik} \right)/S_i^2 \right\rbrace\right]   - \left\lbrace  \log⁡\left( \tilde{\alpha}_j\right)+\log⁡\left( \tilde{\alpha}_k\right)\right\rbrace ,\ &{\rm for}\  1\leq j\neq k \leq p.
	\end{cases}
	\end{equation}
	Under the boundedness condition \ref{conC}, for any positive integer $m$ and $0<\epsilon<3$, there exists a constant $C_0$ depending only on $m$ such that for $1\leq j,k \leq p$,
	\[
	pr\left(\left|\tilde{\sigma}_{jk}-\sigma_{jk}\right|>\epsilon\right)\leq 1/\left( C_0n^m\epsilon^{2m}\right) .
	\]
\end{lemma}

\subsubsection*{Two lemmas for lemma \ref{lem:3} and proofs}
Before prove lemma \ref{lem:3}, We need to prove two additional lemmas first.  
\begin{lemma}\label{lem:1}
	Under the boundedness condition \ref{conC}, for any positive integer $m$, there exists $k_m >0$ such that
	\[
	E(Y_{ij}^m) \le k_m.
	\]
	
\end{lemma}
\begin{proof}[Proof of Lemma~\ref{lem:1}]
	Let $C_{k,m}$ be the Stirling numbers of the second kind. From the moment results of Poisson distribution \citep{Riordan1937}, we have
	\[
	E\left(Y_{ij}^m|X_{ij}\right) = \sum_{k=0}^m \left( S_i X_{ij} \right) ^k C_{k,m}\leq \sum_{k=0}^m  C^k X_{ij}^k C_{k,m}.
	\]
	From the moment generating function (MGF) of the normal distribution, we have
	\[
	E\left( X_{ij}^m \right)= {\rm exp}\left( \frac{1}{2}m^2\sigma_{jj}+m\mu_j\right)\leq {\rm exp}\left( \frac{1}{2}m^2C+mC\right).
	\]
	Combining the above two inequalities, we have
	\begin{equation}\label{lem1}
	\begin{aligned}
	E\left(Y_{ij}^m\right)&=E\left(E\left(Y_{ij}^{m}|X_{ij}\right)\right)\leq E\left(\sum_{k=0}^m  C^k X_{ij}^k C_{k,m}\right)\\
	&=\sum_{k=0}^m C^kE\left( X_{ij}^k \right)  C_{k,m} \leq \sum_{k=0}^m C^k {\rm exp}\left( \frac{1}{2}k^2C+kC\right) C_{k,m}.
	\end{aligned}
	\end{equation}
	Let $k_m=\sum_{k=0}^m C^k {\rm exp}\left( \frac{1}{2}k^2C+kC\right) C_{k,m}$ and then the inequality (\ref{lem1}) leads to Lemma \ref{lem:1}.
\end{proof}

\begin{lemma}\label{lem:2}
	Let $\{W_i, 1\le i \le n\}$ be a series of independent random variables with $E\left(W_i\right)=0$ and $E\left(W_i^{k}\right)\leq u_{k}$ for all $1\leq i\leq n,1\leq k\leq 2m$ where $m$ is a positive integer. Then, there exists a constant $v_m$ only depending on $m$, such that
	$$pr\left(\left|n^{-1}\sum_{i=1}^{n}W_i\right|>\epsilon\right)\leq v_{m}/\left( n^m\epsilon^{2m}\right) .$$
\end{lemma}

\begin{proof}[Proof of Lemma~\ref{lem:2}]
	From Chebyshev inequality, we have
	\begin{equation}\label{i3}
	\begin{aligned}
	pr\left( \left|\frac{1}{n}\sum_{i=1}^{n}W_i\right|>\epsilon\right) &=pr\left(\left(\sum_{i=1}^{n}W_i\right)^{2m}>n^{2m}\epsilon^{2m}\right)
	&\leq E\left(\sum_{i=1}^{n}W_i\right)^{2m}/\left( n^{2m}\epsilon^{2m}\right) .
	\end{aligned}
	\end{equation}
	Combining the inequality (\ref{i3}) and the Rosenthal inequality (\ref{i4}) \citep{Rosenthal1970},
	\begin{equation}\label{i4}
	\begin{aligned}
	E\left(\sum_{i=1}^{n}W_i\right)^{2m}\leq k_m {\rm max}\left[ \sum_{i=1}^{n}E\left(W_i\right)^{2m},\left\lbrace E\left( \sum_{i=1}^{n}W_i^2\right) \right\rbrace ^m\right],
	\end{aligned}
	\end{equation}
	while $k_m$ is a constant that only depends on $m$, we have
	\begin{equation}\label{i5}
	\begin{aligned}
	pr\left(\left|\frac{1}{n}\sum_{i=1}^{n}W_i\right|>\epsilon\right)&\leq  k_m {\rm max}\left[\sum_{i=1}^{n}E\left(W_i\right)^{2m},\left\lbrace E\left( \sum_{i=1}^{n}W_i^2\right)\right\rbrace ^m \right] /\left( n^{2m}\epsilon^{2m}\right) \\
	&\leq k_m {\rm max}\left\{nu_{2m},\left(nu_2\right)^m\right\}/\left( n^{2m}\epsilon^{2m}\right)\\  
	&\leq k_m {\rm max}\left\{u_{2m},u_2^m\right\}/\left( n^{m}\epsilon^{2m}\right) .
	\end{aligned}
	\end{equation}
	Let $v_m=k_m {\rm max}\left\{u_{2m},u_2^m\right\}$ in the inequality (\ref{i5}) then Lemma \ref{lem:2} follows.
\end{proof}
\subsubsection*{Proof of lemma \ref{lem:3}}
\begin{proof}
	For any $1\leq j,k\leq p$, notice that 
	\[
	\alpha_j={\rm exp}\left( \mu_j+\sigma_{jj}/2\right)=E\left( Y_{ij}/S_i\right), \tilde{\alpha}_j=n^{-1} \sum_{i=1}^n Y_{ij}/S_i,
	\]
	and let 
	\[
	u_j=\alpha_j^2{\rm exp}\left( \sigma_{jj}\right)=E\left(\left( Y_{ij}^2-Y_{ij}\right) /S_i^2\right) , \tilde{u}_j=n^{-1}\sum_{i=1}^{n}\left( \left( Y_{ij}^2-Y_{ij}\right) / S_i^2 \right) , 
	\]
	\[
	w_{jk}=\alpha_j\alpha_{k}{\rm exp}\left( \sigma_{jk}\right)=E\left(Y_{ij}Y_{ik}/S_i^2\right), \tilde{w}_{jk}=n^{-1}\sum_{i=1}^{n}\left( Y_{ij}Y_{ik}/S_i^2\right).
	\]
	
	Because $\left[ Y_{ij}\right] _{1\leq i\leq n},\left[ \left( Y_{ij}^2-Y_{ij}\right) /S_i^2\right] _{1\leq i\leq n},\left[ Y_{ij}Y_{ik}/S_i^2\right] _{1\leq i\leq n}$ are three sets of independent variables, all of which have finite $m$th moments for any positive integer $m$ by Lemma \ref{lem:1}. Then, by Lemma \ref{lem:2}, we have
	\[
	pr\left(\left|\tilde{\alpha}_j-\alpha_j\right|>\epsilon\right)\leq \frac{v_{1m}}{n^m\epsilon^{2m}},
	pr\left(\left|\tilde{u}_j-u_j\right|>\epsilon\right)\leq \frac{v_{2m}}{n^m\epsilon^{2m}},
	pr\left(\left|\tilde{w}_{jk}-w_{jk}\right|>\epsilon\right)\leq \frac{v_{3m}}{n^m\epsilon^{2m}}.
	\]
	
	Now we can derive the convergence rate of $\tilde{\sigma}_{jk}$. Using the boundedness condition (\ref{conC}), the parameters $\alpha_j,u_j,w_{jk}$ are all in the interval $\left[ {\rm exp}\left(-3C \right) ,{\rm exp}\left(4C \right)\right] $. Then, for any $\epsilon<{\rm exp}\left(-3C \right)/2$, we have
	\begin{equation}\label{lemma3:probineq}
	pr\left( \max\left\lbrace \left|\tilde{\alpha}_j-\alpha_j\right|,~ \left|\tilde{u}_j-u_j\right|,~\left|\tilde{w}_{ij'}-w_{ij'}\right|\right\rbrace \leq\epsilon \right) > 1-\left( v_{1m}+v_{2m}+v_{3m}\right) /\left( n^m\epsilon^{2m}\right).~
	\end{equation}
	Then with at least probability $1-\left( v_{1m}+v_{2m}+v_{3m}\right) /\left( n^m\epsilon^{2m}\right) $, 
	\begin{equation}\label{lemma3:a}
	\max\left\lbrace \left|\tilde{\alpha}_j-\alpha_j\right|,~ \left|\tilde{u}_j-u_j\right|,~\left|\tilde{w}_{ij'}-w_{ij'}\right|\right\rbrace \leq\epsilon,
	\end{equation}
	according to $\alpha_j,u_j,w_{jk}\geq{\rm exp}\left(-3C \right)$ and $\epsilon<{\rm exp}\left(-3C \right)/2$ , we can derive from (\ref{lemma3:a}) that
	\begin{equation}\label{lemma3:b}
	\min\left\lbrace \tilde{\alpha}_j,\tilde{u}_j,\tilde{w}_{jk}\right\rbrace >{\rm exp}\left(-3C \right)/2. 
	\end{equation}
	For any $j\neq k$,
	\begin{equation}\label{lemma3:u}
	\sigma_{jj}=\log u_j-2\log\alpha_j,~
	\sigma_{jk}=\log w_{jk}-\log\alpha_j-\log\alpha_{k}. 
	\end{equation}
	From the Lagrange's mean value theorem, we have, for any $ x,y\geq{\rm exp}\left(-3C \right)/2$,
	\begin{equation}\label{lemma3:c}
	\left| \log x-\log y\right|=\left|x-y\right|/\xi \leq2\left|x-y\right|/{\rm exp}\left(-3C \right),
	\end{equation}
	while $\xi$ is a number between $x,y$. 
	Then combining (\ref{lemma3:a}), (\ref{lemma3:b}) and (\ref{lemma3:c}), we have 
	\[
	\max \left\lbrace \left| \log\tilde{\alpha}_j-\log\alpha_j\right|,\left| \log\tilde{u}_j-\log u_j\right|,\left| \log\tilde{w}_{jk}-\log w_{jk}\right| \right\rbrace  \leq 2{\rm exp}\left(3C \right)\epsilon,~
	\]
	and thus $\left|\tilde{\sigma}_{jk}-\sigma_{jk}\right|\leq 6{\rm exp}\left(3C \right)\epsilon$ using (\ref{lemma3:u}).~Then from the probability inequality \eqref{lemma3:probineq}, for any $\epsilon<{\rm exp}\left(-3C \right)/2$, we have
	\[
	pr\left( \left|\tilde{\sigma}_{jk}-\sigma_{jk}\right|\leq 6{\rm exp}\left(3C \right)\epsilon \right)> 1-\frac{v_{1m}+v_{2m}+v_{3m}}{n^m\epsilon^{2m}}.~
	\]
	So, for any $\eta=6{\rm exp}\left(3C \right)\epsilon<3$ and $C_0= \left\lbrace 6{\rm exp}\left(3C \right)\right\rbrace ^{-2m}\left(v_{1m}+v_{2m}+v_{3m}\right)^{-1}$, we have
	\[
	pr\left(\left|\tilde{\sigma}_{jk}-\sigma_{jk}\right|>\eta\right)\leq \left\lbrace 6{\rm exp}\left(3C \right)\right\rbrace ^{2m}\frac{v_{1m}+v_{2m}+v_{3m}}{n^m\eta^{2m}} = 1/\left( C_0 n^m\eta^{2m}\right) .~
	\]
	Then we finish the proof of lemma \ref{lem:3}
\end{proof}
\subsubsection{Proofs of theorems}
\subsubsection*{Proof of Theorem \ref{thm:rate}}
According to Lemma \ref{lem:3}, for any $0<\eta<3$,  $pr\left(\left|\tilde{\sigma}_{jk}-\sigma_{jk}\right|>\eta\right)\leq 1/\left( C_0n^m\eta^{2m}\right) $, we have
\[
pr\left(\left\|\tilde{\Sigma}-\Sigma\right\|_{\infty}>\epsilon\right)\leq p^2/\left( C_0n^m\epsilon^{2m}\right) .~
\]
Then, from
\[
\check{\Sigma}=\argmin_{A\succeq 0}\left\|A-\tilde{\Sigma}\right\|_{\infty}\Rightarrow \left\|\check{\Sigma}-\tilde{\Sigma}\right\|_{\infty}\leq\left\|\Sigma -\tilde{\Sigma}\right\|_{\infty},
\]
we have
\[
\left\|\hat{\Sigma}-\Sigma\right\|_{\infty}\leq\left\|\hat{\Sigma}-\check{\Sigma}\right\|_{\infty}+\left\|\check{\Sigma}-\tilde{\Sigma}\right\|_{\infty}+\left\|\tilde{\Sigma}-\Sigma\right\|_{\infty}\]\[
=2\left\|\check{\Sigma}-\tilde{\Sigma}\right\|_{\infty}+\left\|\tilde{\Sigma}-\Sigma\right\|_{\infty} \leq 3\left\|\tilde{\Sigma}-\Sigma\right\|_{\infty},~
\]
and 
\[
\left\|\check{\Sigma}-\Sigma\right\|_{\infty}\leq 2\left\|\tilde{\Sigma}-\Sigma\right\|_{\infty}.~
\]
Let $C_1=3^{-2m}C_0$, $C_2=2^{-2m}C_0$ and for any $0 < \epsilon  < 6$, we have~
\[
pr\left(\left\|\hat{\Sigma}-\Sigma\right\|_{\infty}>\epsilon \right) \le pr\left(\left\|\tilde{\Sigma}-\Sigma\right\|_{\infty}>\epsilon/3 \right) \leq p^2/\left( C_1 n^m \epsilon^{2m}\right) .~
\]
\[
pr\left(\left\|\check{\Sigma}-\Sigma\right\|_{\infty}>\epsilon \right) \le pr\left(\left\|\tilde{\Sigma}-\Sigma\right\|_{\infty}>\epsilon/2 \right) \leq p^2/\left( C_2 n^m \epsilon^{2m}\right) .~
\]
\vspace*{-10pt}

\subsubsection*{Proof of Theorem \ref{thm:recovery}, \ref{thm:sign} and \ref{thm:both}}
We define
$$\breve{\Theta}=\argmin_{A=A^T}\frac{1}{2}{\rm tr}\left(\hat{\Sigma}A^2\right)-{\rm tr}\left(A\right)+\lambda\left\|A\right\|_{1,\text{off}}. $$
Let 
\begin{eqnarray*}
	\epsilon=&1/{\rm max}\Bigg[12dk_{\Gamma},~12\gamma^{-1}(k_{\Sigma}k_{\Gamma}^2+k_{\Gamma}),~\left\lbrace 12\gamma^{-1}\left(k_{\Sigma}k_{\Gamma}^3+k_{\Gamma}^2\right)+5dk_{\Gamma}^2\right\rbrace \theta_{\min}^{-1},\\
	&{\rm min}\left\lbrace s^{1/2},d+1\right\rbrace \left\lbrace 12\gamma^{-1}\left(k_{\Sigma}k_{\Gamma}^3+k_{\Gamma}^2\right)+5dk_{\Gamma}^2\right\rbrace \lambda_{\min}^{-1}(\Theta),1/5\Bigg]
\end{eqnarray*}
For $\eta>2$, let $n_f=C_1^{-1/m}p^{\eta/m}\epsilon^{-2}$ and $\epsilon_f=C_1^{-1/\left( 2m\right) }p^{\eta/\left( 2m\right) }n^{-1/2}$. According to $n>n_f$, we have $$\epsilon_f=C_1^{-1/\left( 2m\right) }p^{\eta/\left( 2m\right)}n^{-1/2}<C_1^{-1/\left( 2m\right)}p^{\eta/\left( 2m\right)}n_f^{-1/2}=\epsilon<6,$$
while $C_1$ is the constant in Theorem \ref{thm:rate}. Then, from Theorem \ref{thm:rate}, we have
\[
pr\left(\|\hat{\Sigma} - \Sigma \|_{\infty}
>\epsilon_f\right)<p^2/\left( C_1 n^m \epsilon_f^{2m}\right) =p^{2-\eta},
\]
and thus with a probability at least $1-p^{2-\eta}$, $$\|\hat{\Sigma} - \Sigma \|_{\infty}\leq\epsilon_f< \epsilon\leq 1/{\rm max}\left\{12dk_{\Gamma},12\gamma^{-1}\left(k_{\Sigma}k_{\Gamma}^2+k_{\Gamma}\right)\right\}.$$
According to $\lambda=12\gamma^{-1}\left(k_{\Sigma}k_{\Gamma}^2+k_{\Gamma}\right)\epsilon_f$,
we can get that
\begin{equation}\label{ref10}
\begin{aligned}
\|\hat{\Sigma} - \Sigma \|_{\infty}&<1/\left( 12dk_{\Gamma}\right) ,\\
6\|\hat{\Sigma} - \Sigma \|_{\infty}\left(k_{\Sigma}k_{\Gamma}^2+k_{\Gamma}\right)&\leq 0.5\gamma {\rm min}\left\{\lambda,1\right\}.
\end{aligned}
\end{equation}
Using Lemma \ref{dtrace} (a) with (\ref{ref10}), $\breve{\Theta}$ recovers all zeros in $\Theta$. That is
$$\text{vec}\left(\breve{\Theta}\right)_{G^c}=0.$$
\noindent
Using Lemma \ref{dtrace} (b) and according to the fact that $$\lambda=12\gamma^{-1}\left(k_{\Sigma}k_{\Gamma}^2+k_{\Gamma}\right)\epsilon_f< 12\gamma^{-1}\left(k_{\Sigma}k_{\Gamma}^2+k_{\Gamma}\right)\epsilon\leq 1$$ and $\|\hat{\Sigma} - \Sigma \|_{\infty}\leq\epsilon_f$, we have
\begin{equation}\label{8}
\begin{aligned}
\left\|\breve{\Theta}-\Theta\right\|_{\infty}&<\lambda k_{\Gamma}+\frac{5}{2}d\left(1+\lambda\right)\|\hat{\Sigma} - \Sigma \|_{\infty} k_{\Gamma}^2\\
&\leq \left\lbrace 12\gamma^{-1}\left(k_{\Sigma}k_{\Gamma}^3+k_{\Gamma}^2\right)+5dk_{\Gamma}^2\right\rbrace \epsilon_f.~
\end{aligned}
\end{equation}
Then we consider the $s$ nonzeros in $\Theta$ and $\text{vec}\left(\breve{\Theta}\right)_{G^c}=0$, we can easily get
\begin{equation}\label{9}
\begin{aligned}
\left\|\breve{\Theta}-\Theta\right\|_F&\leq s^{1/2}\left\|\breve{\Theta}-\Theta\right\|_{\infty}\\
&<s^{1/2}\left\lbrace 12\gamma^{-1}\left(k_{\Sigma}k_{\Gamma}^3+k_{\Gamma}^2\right)+5dk_{\Gamma}^2\right\rbrace \epsilon_f.\\
\end{aligned}
\end{equation}
Using $\|A\|_2\leq\|A\|_F$ and $\|A\|_2\succeq 0$ while $\left| A_{jj}\right| \geq \sum_{k\neq j}\left| A_{jk}\right|$ for all $1\leq j\leq p$
\begin{equation}\label{10}
\begin{aligned}
\left\|\breve{\Theta}-\Theta\right\|_2&\leq {\rm min}\left\{s^{1/2},d+1\right\}\left\|\breve{\Theta}-\Theta\right\|_{\infty}\\
&<{\rm min}\left\{s^{1/2},d+1\right\}\left\lbrace 12\gamma^{-1}\left(k_{\Sigma}k_{\Gamma}^3+k_{\Gamma}^2\right)+5dk_{\Gamma}^2\right\rbrace \epsilon_f.\\
\end{aligned}
\end{equation}
From (\ref{8}) and combining $\epsilon_f<\epsilon\leq \theta_{\min}/\left(12\gamma^{-1}\left(k_{\Sigma}k_{\Gamma}^3+k_{\Gamma}^2\right)+5dk_{\Gamma}^2\right)$, we have
$$\left\|\breve{\Theta}-\Theta\right\|_{\infty}<\theta_{\min},$$
which means that $\breve{\Theta}$ also recovers the nonzeros in $\Theta$.

Finally, we check $\hat{\Theta}=\breve{\Theta}$ to finish the proof. We just need to verify $\lambda_{\min}\left(\breve{\Theta}\right)>0$, that can be obtained from $\left\|\breve{\Theta}-\Theta\right\|_2<\lambda_{\min}\left(\Theta\right)$. So using (\ref{10}) and combining $$\epsilon_f<\epsilon\leq \lambda_{\min}\left(\Theta\right)/\left[ {\rm min}\left\{s^{1/2},d+1\right\}\left\lbrace 12\gamma^{-1}\left(k_{\Sigma}k_{\Gamma}^3+k_{\Gamma}^2\right)+5dk_{\Gamma}^2\right\rbrace \right] ,$$ we get the conclusion.

Above all, $\hat{\Theta}$ recovers all zeros and nonzeros in $\Theta$ and meet all the convergence rates for $\breve{\Theta}$ in (\ref{8}), (\ref{9}), (\ref{10}) with a probability at least $1-p^{2-\eta}$, then we finish the proof of Theorem \ref{thm:recovery} and \ref{thm:sign}. 

Replace $\hat{\Theta},C_1$ with $\check{\Theta},C_2$, then the proof of the Theorem \ref{thm:both} is the same as the proof above.

\subsection{Additional results of simulation and real data analysis}
\begin{table}[h]
	\centering
	\caption{Comparisons of PLNet with VPLN and glasso in terms of the area under Receiver Operating Characteristic curve (AUC) on simulation results. The results are averages over 100 replicates with standard deviations in brackets}
	\includegraphics[scale=0.8]{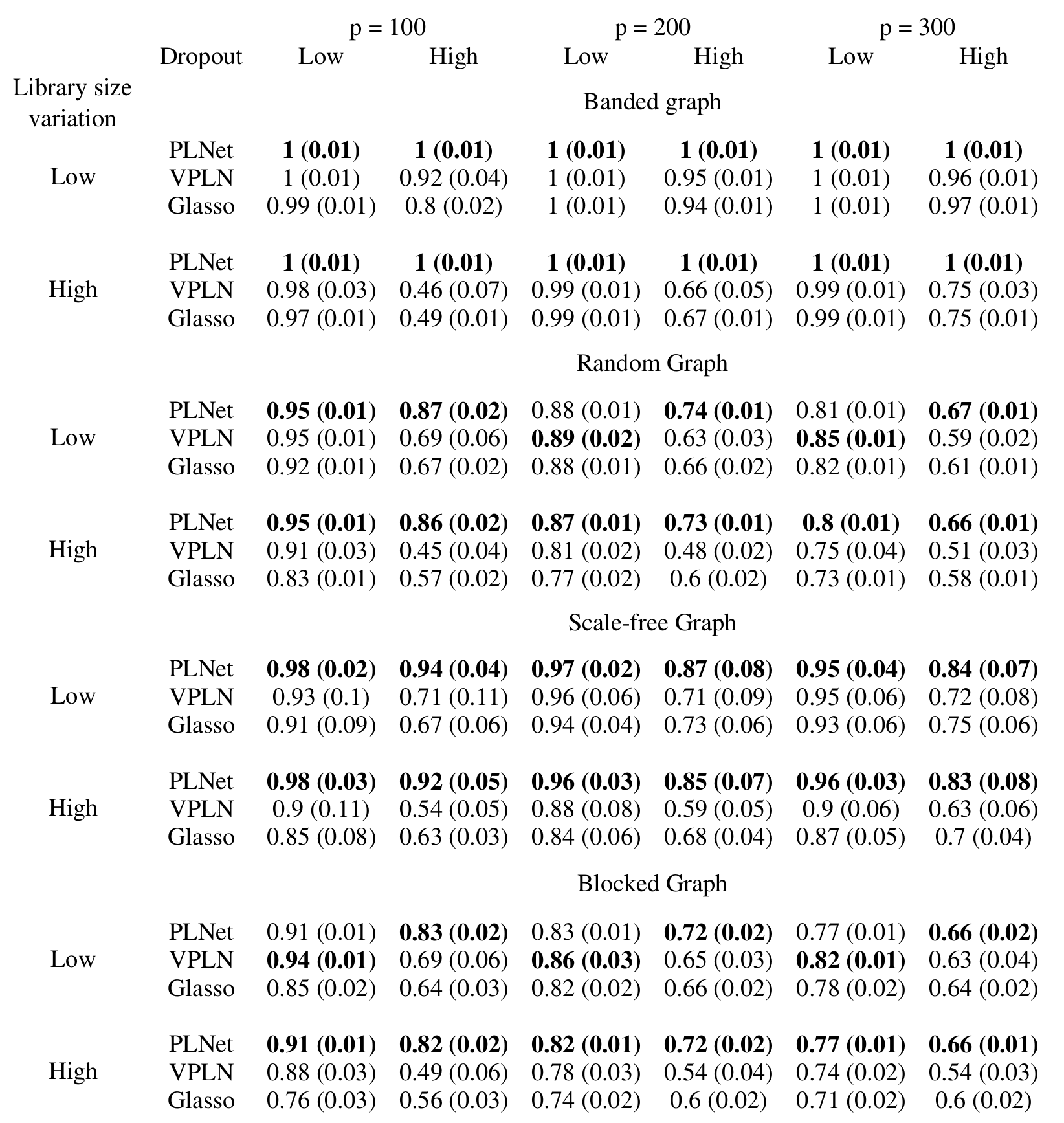} \label{4}
\end{table}

\begin{table}[htb]
	\centering
	\caption{Comparisons of PLNet with VPLN and glasso in terms of the true positive rate (TPR), the true discovery rate (TDR) and the Frobenius risk for band graph. The tuning parameters of three methods are selected by BIC criterion. The results are averages over 100 replicates with standard deviations in brackets}
	\includegraphics[scale=0.9]{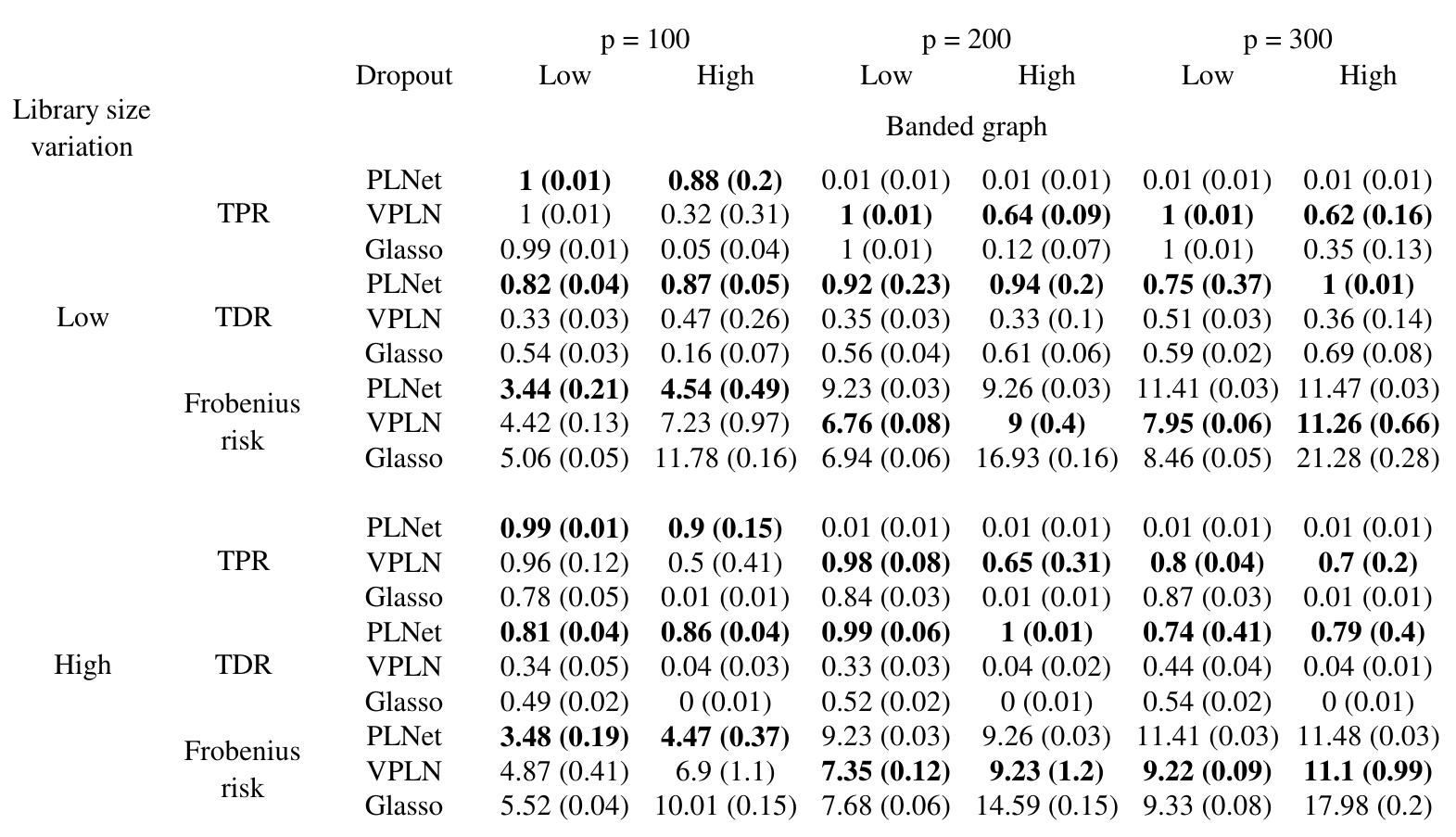} \label{5}
\end{table}

\begin{table}[htb]
	\centering
	\caption{Comparisons of PLNet with VPLN and glasso in terms of the true positive rate (TPR), the true discovery rate (TDR) and the Frobenius risk for Scale-free graph. The tuning parameters of three methods are selected by BIC criterion. The results are averages over 100 replicates with standard deviations in brackets}
	\includegraphics[scale=0.9]{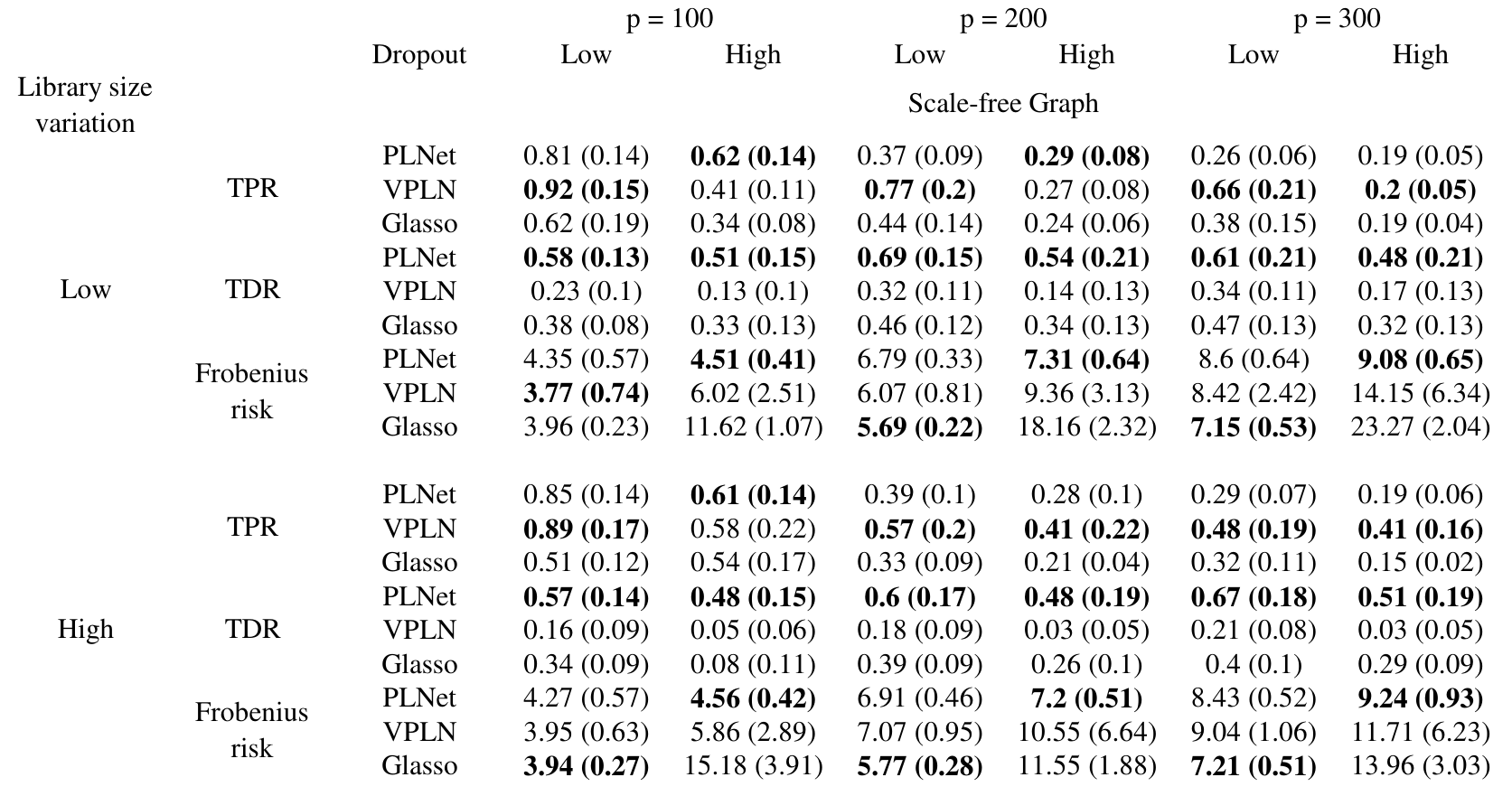} \label{6}
\end{table}

\begin{table}[htb]
	\centering
	\caption{Comparisons of PLNet with VPLN and glasso in terms of the true positive rate (TPR), the true discovery rate (TDR) and the Frobenius risk for block graph. The tuning parameters of three methods are selected by BIC criterion. The results are averages over 100 replicates with standard deviations in brackets}
	\includegraphics[scale=0.9]{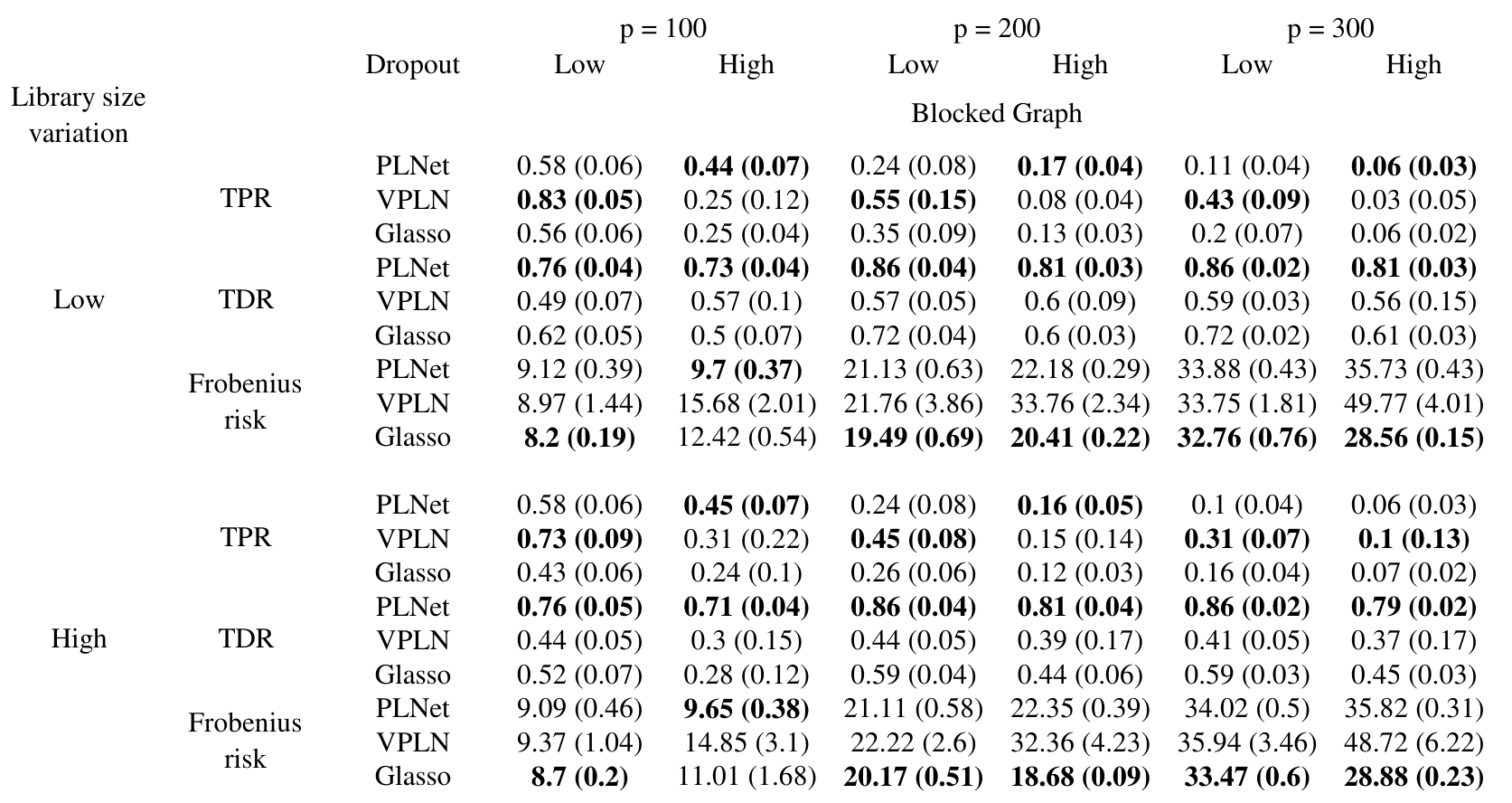} \label{7}
\end{table}

\begin{figure}[htb]
	\centering
	\subfigure[PLNet]{
		\begin{minipage}[b]{0.5\linewidth}
			\includegraphics[width=2.8in]{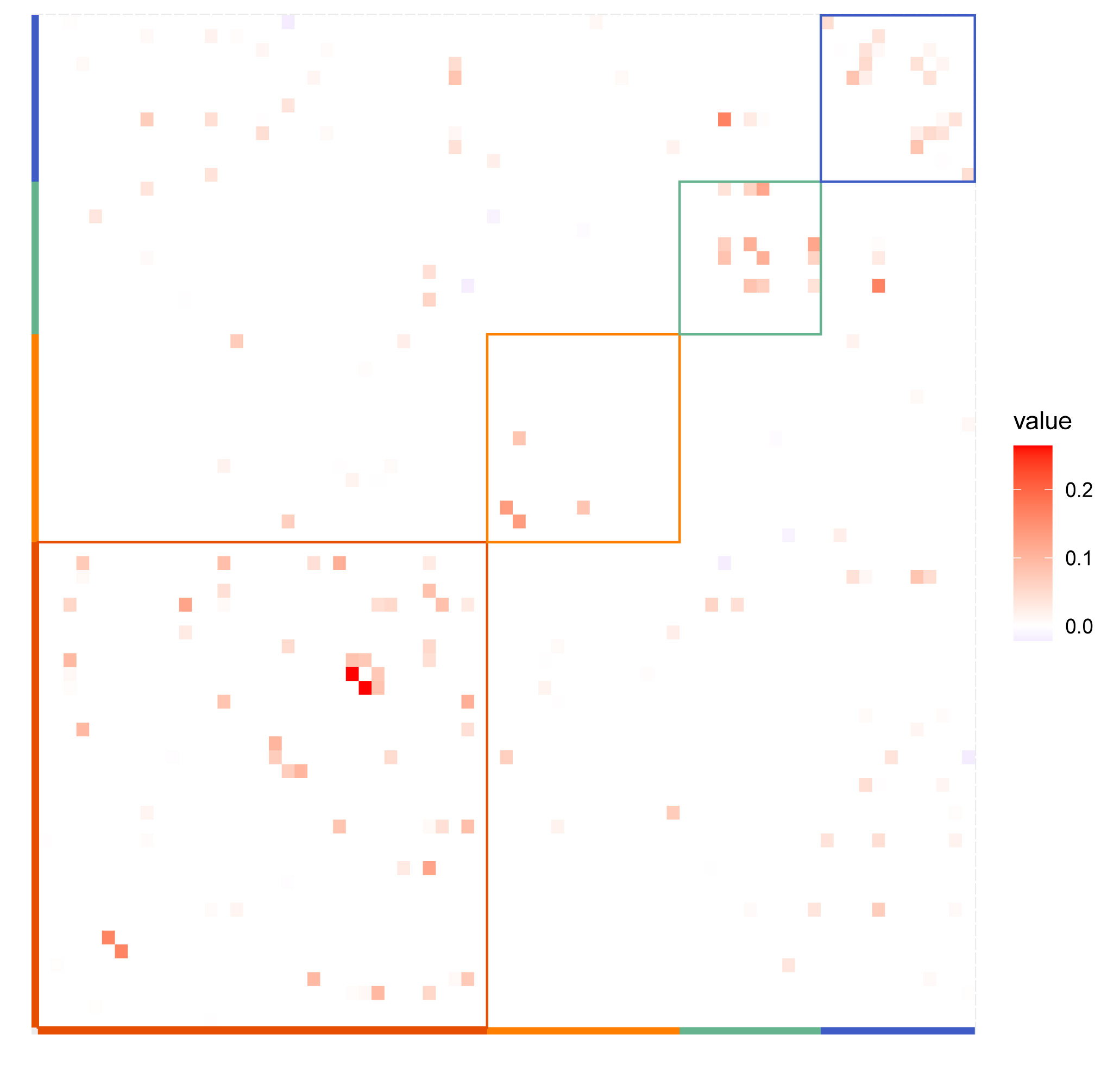}
		\end{minipage}%
	}%
	\subfigure[VPLN]{
		\begin{minipage}[b]{0.5\linewidth}
			\includegraphics[width=2.8in]{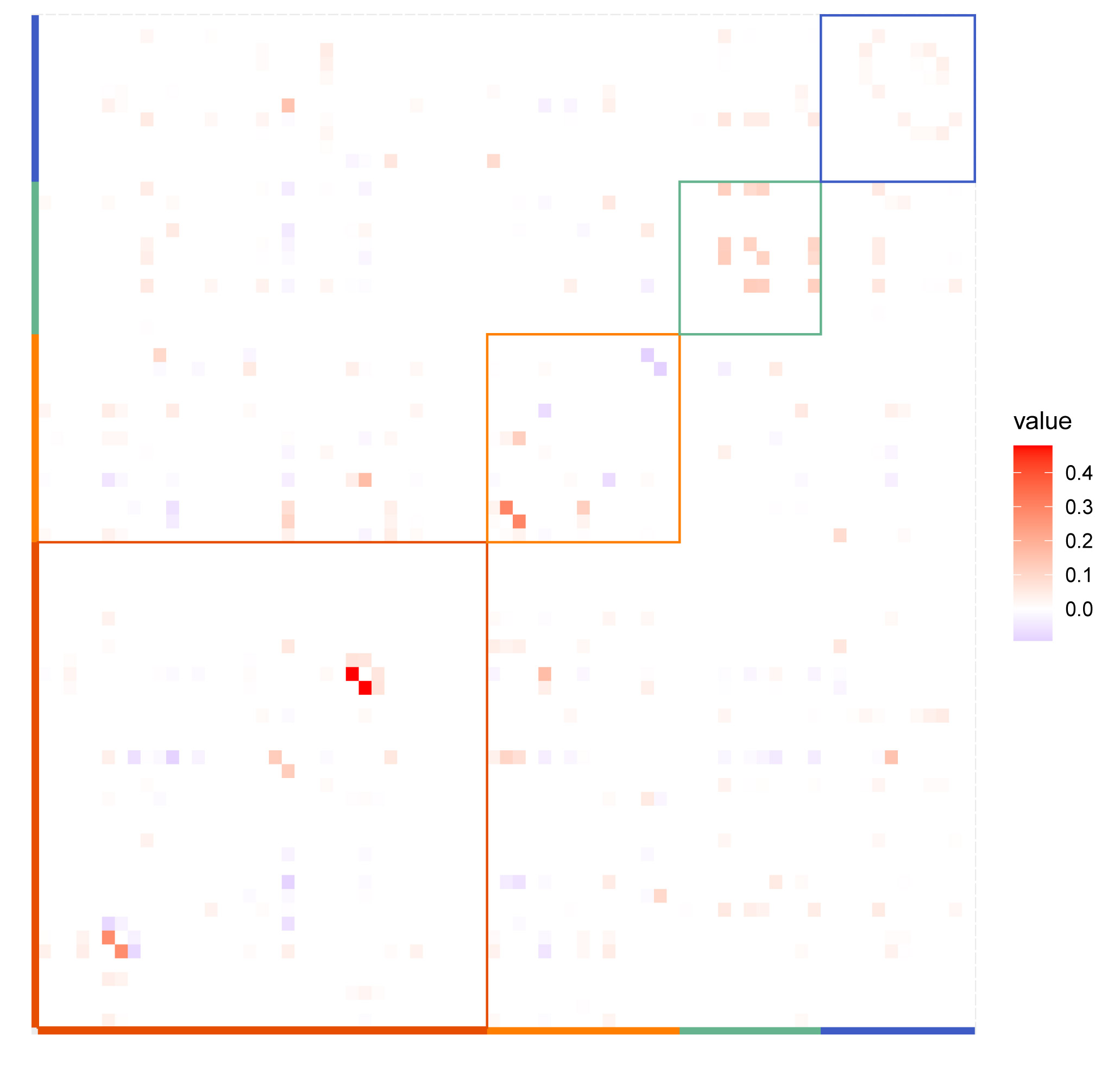}
		\end{minipage}%
	}%
	\caption{Heat maps of partial correlations between genes in the 4 GO modules given by PLNet (a) and VPLN (b) tuned such that the network densities are around 3\%. Red: cytokine-mediated signaling pathway (Module $M_1$); Yellow: neutrophil mediated immunity (Module $M_2$); Green: cellular protein metabolic process (Module $M_3$); Blue: proteolysis (Module $M_4$)}
	\label{fig:heatmap11}
\end{figure}
\begin{figure}[htb]
	\centering
	\subfigure[PLNet]{
		\begin{minipage}[t]{0.5\linewidth}
			\centering
			\includegraphics[width=2.8in]{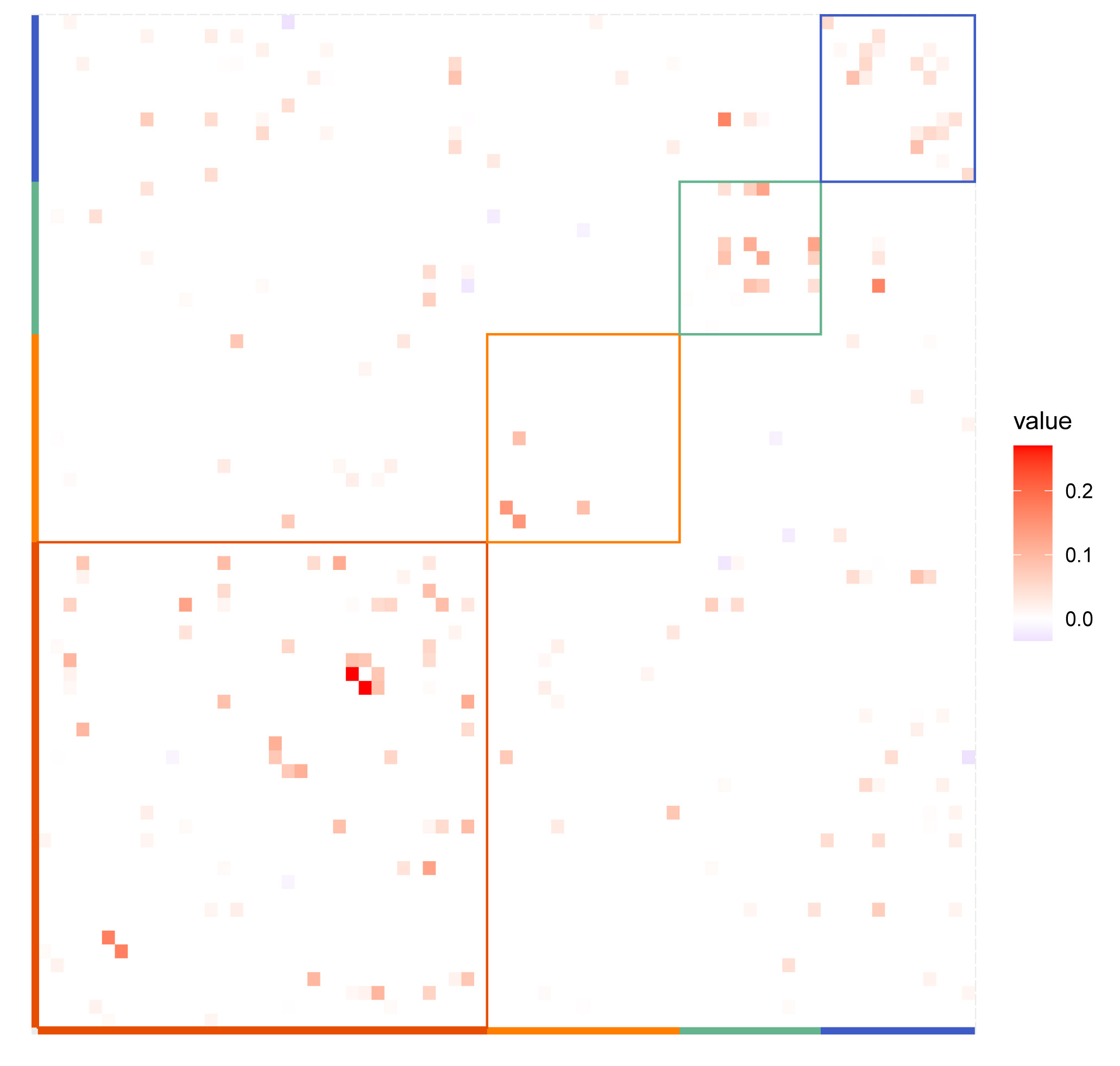}
		\end{minipage}%
	}%
	\subfigure[VPLN]{
		\begin{minipage}[t]{0.5\linewidth}
			\centering
			\includegraphics[width=2.8in]{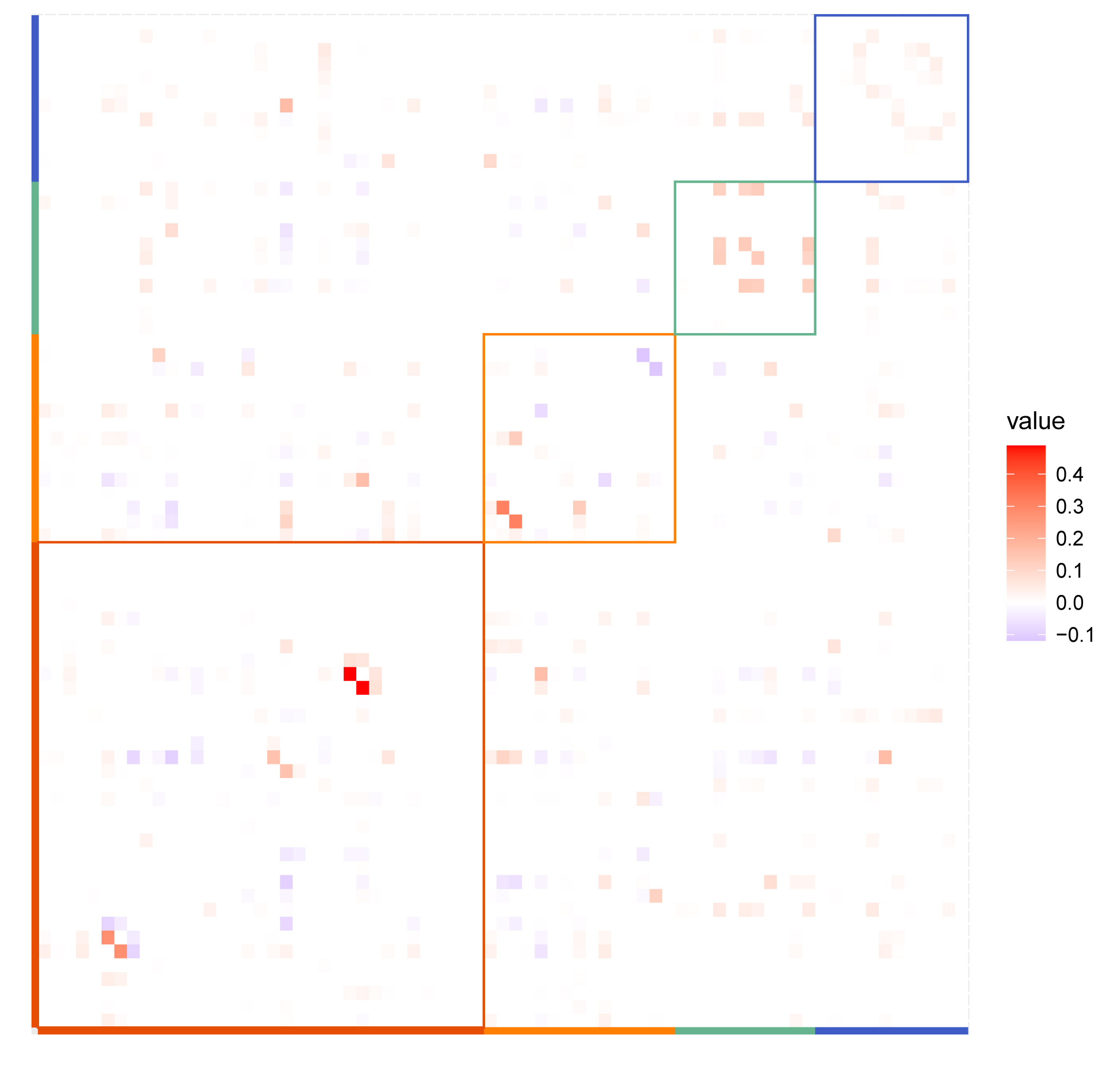}
		\end{minipage}%
	}%
	\caption{Heat maps of partial correlations between genes in the 4 GO modules given by PLNet (a) and VPLN (b) tuned such that the network densities are around 4\%. Red: cytokine-mediated signaling pathway (Module $M_1$); Yellow: neutrophil mediated immunity (Module $M_2$); Green: cellular protein metabolic process (Module $M_3$); Blue: proteolysis (Module $M_4$)}
	\label{fig:heatmap2}
\end{figure}
\begin{figure}[htb]
	\centering
	\subfigure[PLNet]{
		\begin{minipage}[t]{0.5\linewidth}
			\centering
			\includegraphics[width=2.8in]{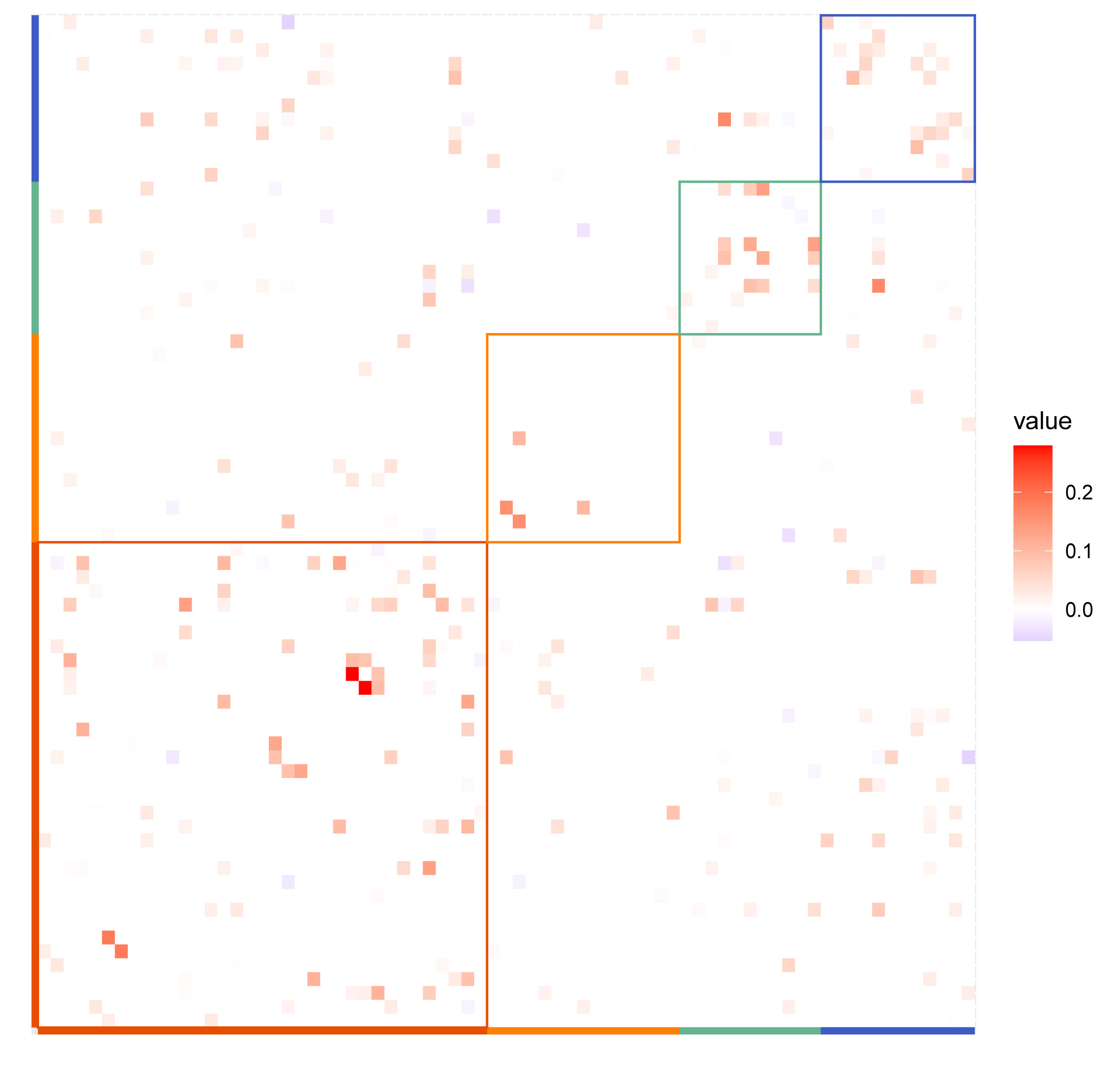}
		\end{minipage}%
	}%
	\subfigure[VPLN]{
		\begin{minipage}[t]{0.5\linewidth}
			\centering
			\includegraphics[width=2.8in]{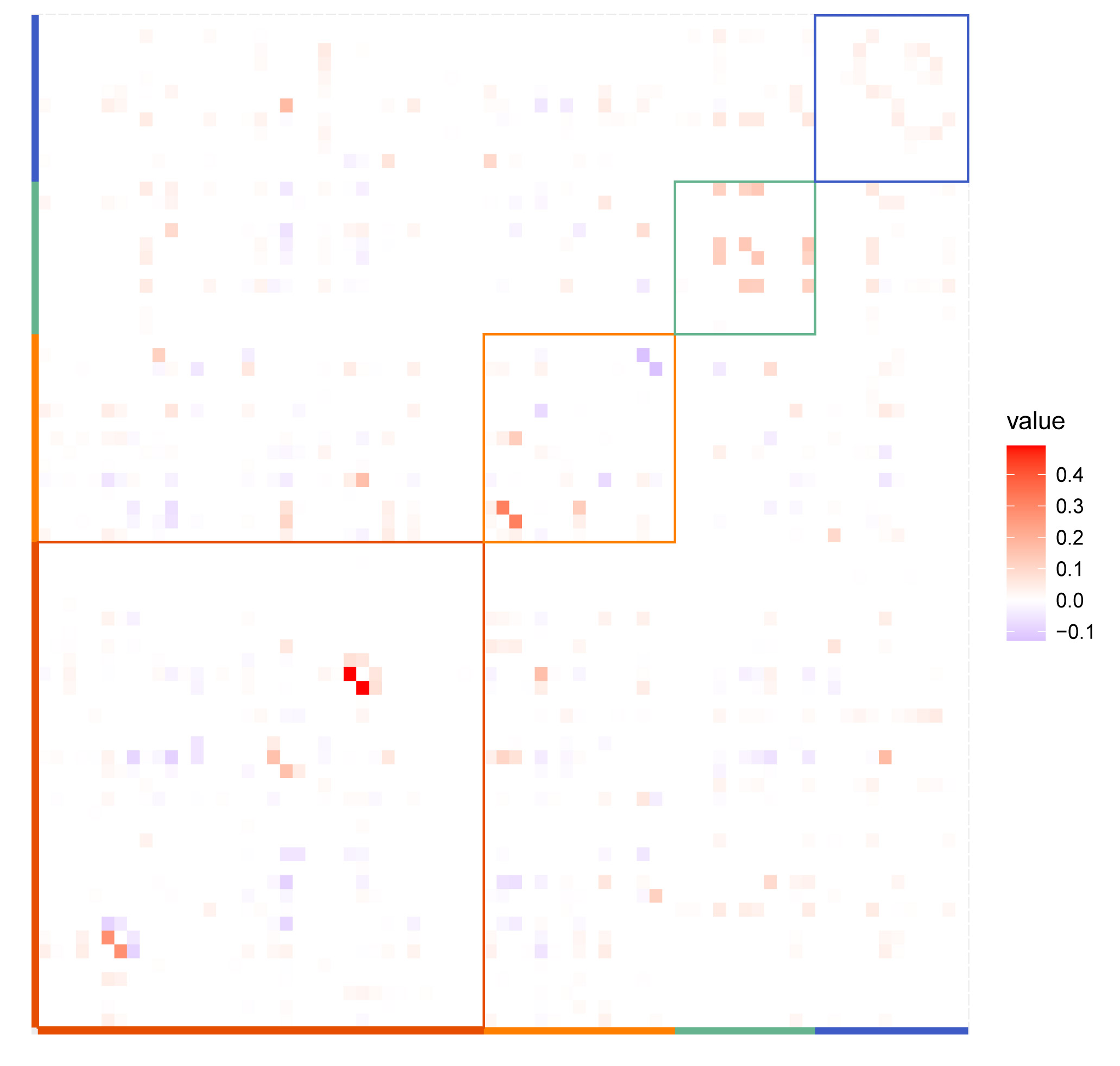}
		\end{minipage}%
	}%
	\caption{Heat maps of partial correlations between genes in the 4 GO modules given by PLNet (a) and VPLN (b) tuned such that the network densities are around 6\%. Red: cytokine-mediated signaling pathway (Module $M_1$); Yellow: neutrophil mediated immunity (Module $M_2$); Green: cellular protein metabolic process (Module $M_3$); Blue: proteolysis (Module $M_4$)}
	\label{fig:heatmap3}
\end{figure}
\begin{figure}[htb]
	\centering
	\subfigure[PLNet]{
		\begin{minipage}[t]{0.5\linewidth}
			\centering
			\includegraphics[width=2.8in]{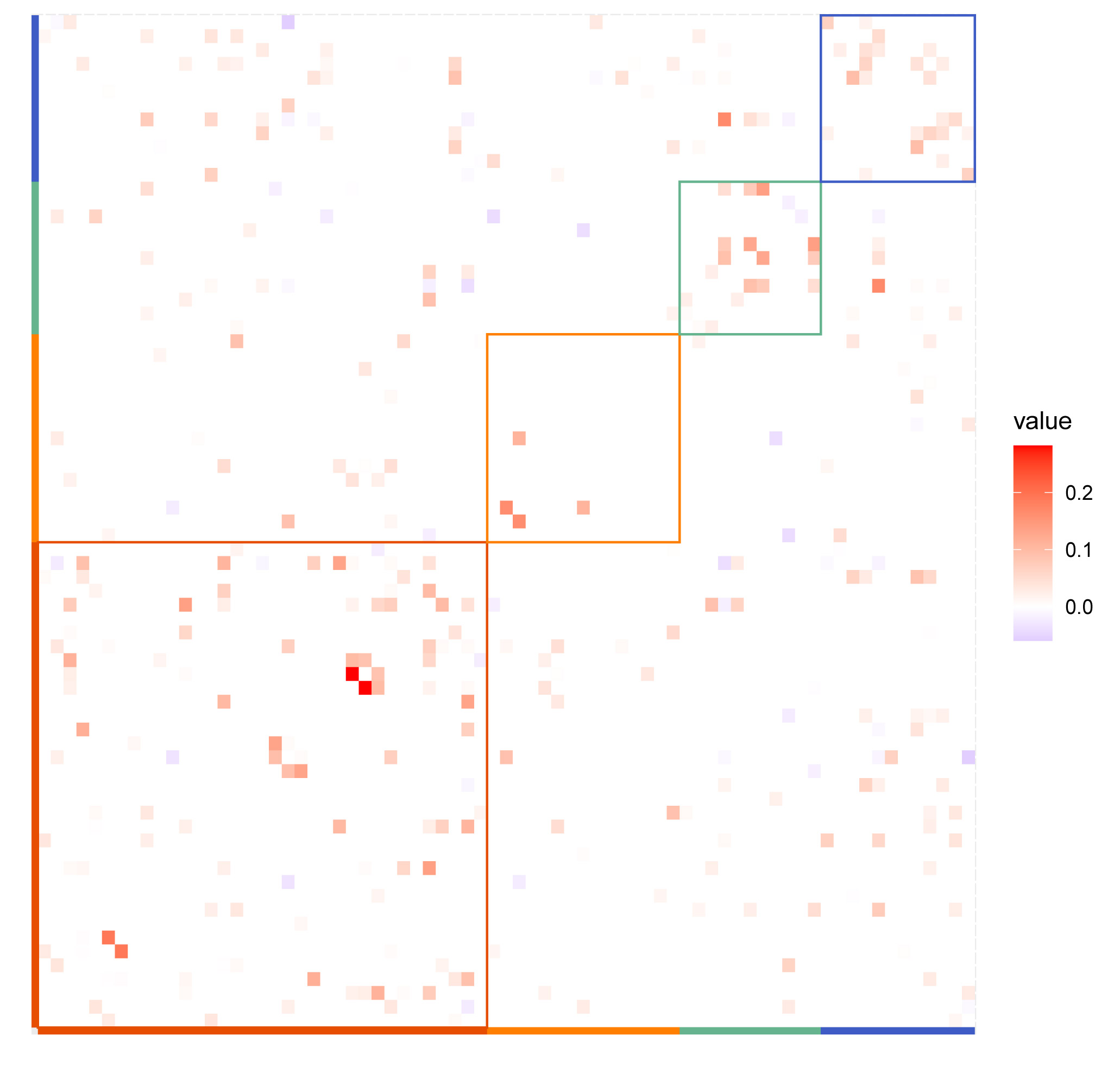}
		\end{minipage}%
	}%
	\subfigure[VPLN]{
		\begin{minipage}[t]{0.5\linewidth}
			\centering
			\includegraphics[width=2.8in]{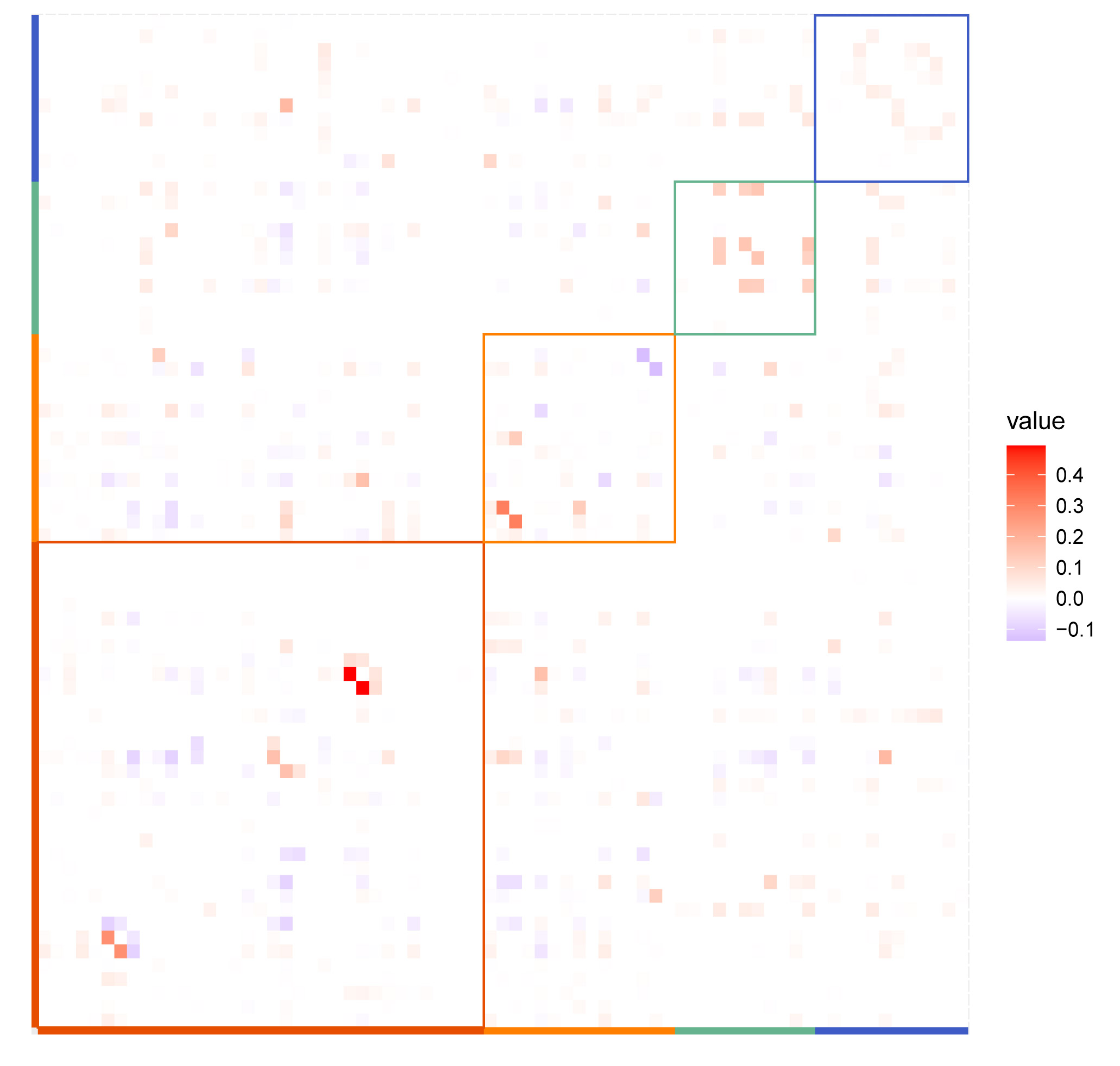}
		\end{minipage}%
	}%
	\caption{Heat maps of partial correlations between genes in the 4 GO modules given by PLNet (a) and VPLN (b) tuned such that the network densities are around 7\%. Red: cytokine-mediated signaling pathway (Module $M_1$); Yellow: neutrophil mediated immunity (Module $M_2$); Green: cellular protein metabolic process (Module $M_3$); Blue: proteolysis (Module $M_4$)}
	\label{fig:heatmap4}
\end{figure}
\begin{figure}[htb]
	\centering
	\subfigure[PLNet]{
		\begin{minipage}[t]{0.5\linewidth}
			\centering
			\includegraphics[width=2.8in]{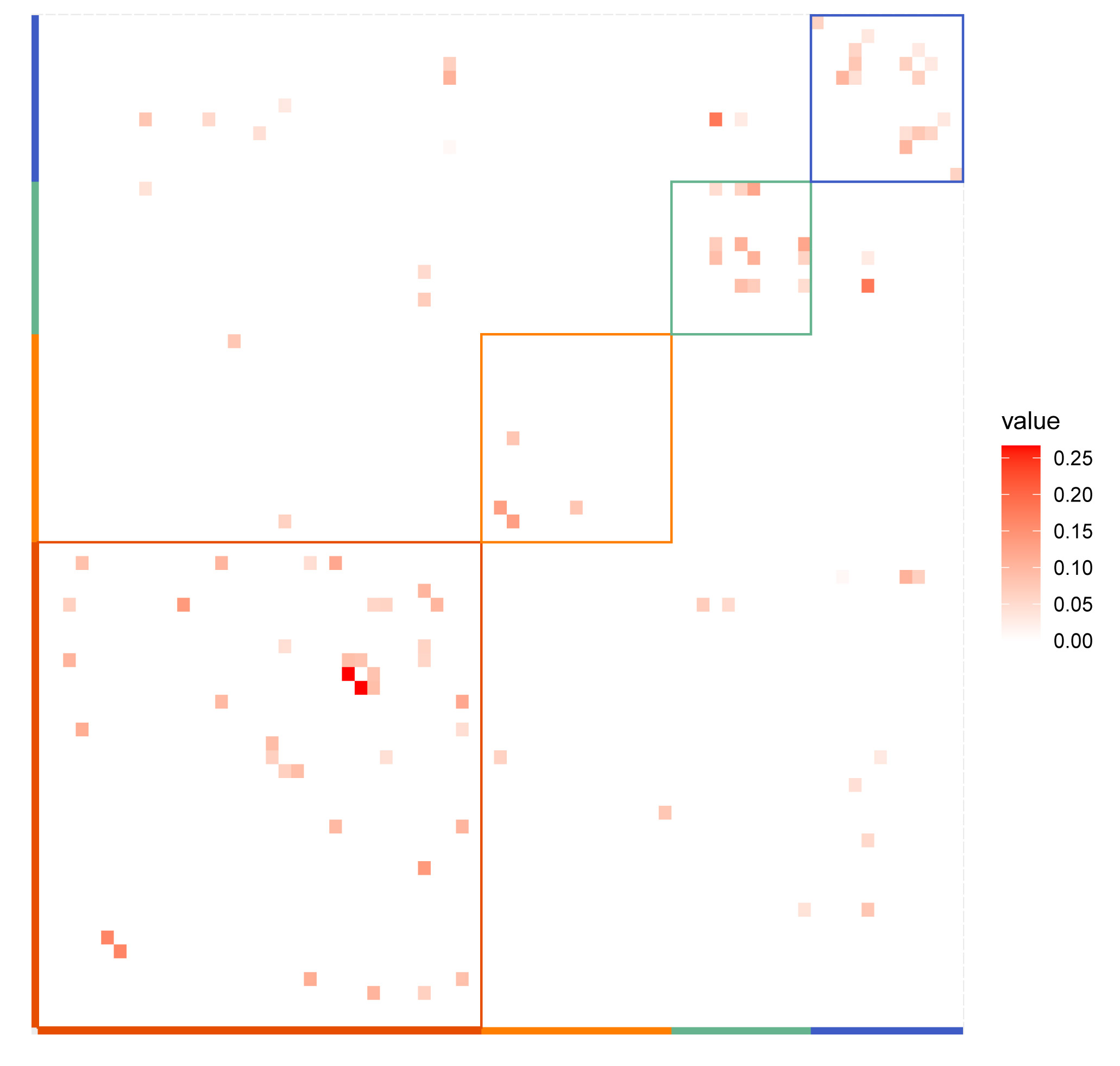}
		\end{minipage}%
	}%
	\subfigure[VPLN]{
		\begin{minipage}[t]{0.5\linewidth}
			\centering
			\includegraphics[width=2.8in]{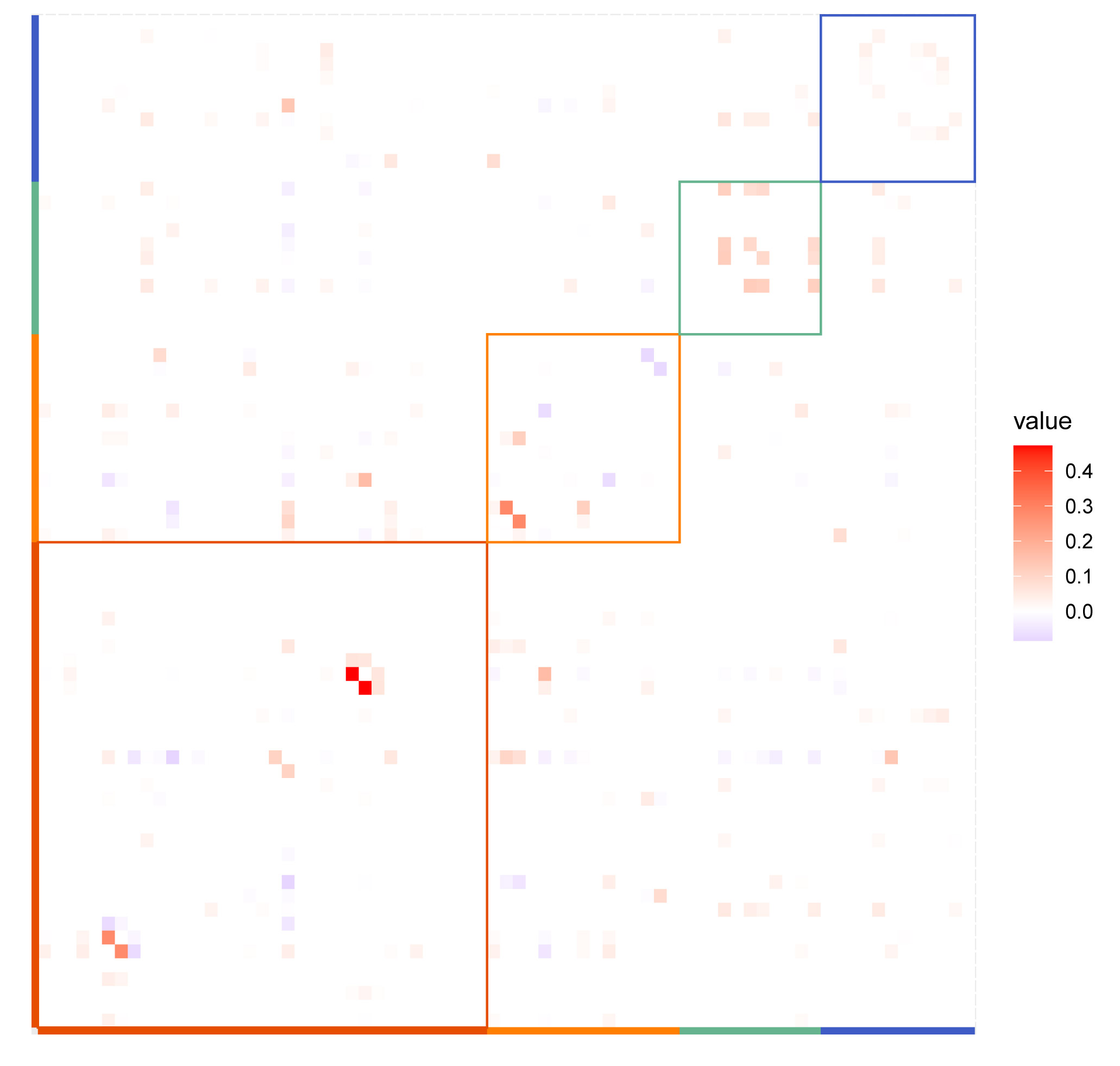}
		\end{minipage}%
	}%
	\caption{Heat maps of partial correlations between genes in the 4 GO modules given by PLNet (a) and VPLN (b) tuned by BIC. Red: cytokine-mediated signaling pathway (Module $M_1$); Yellow: neutrophil mediated immunity (Module $M_2$); Green: cellular protein metabolic process (Module $M_3$); Blue: proteolysis (Module $M_4$)}
	\label{fig:heatmap5}
\end{figure}
\begin{table}[htb]
	\centering
	\caption{The within-between connection ratios of the 4 major modules in the networks estimated by PLNet and VPLN tuned such that the network densities are around 3-7\% and by BIC}
	\includegraphics[scale=1.1]{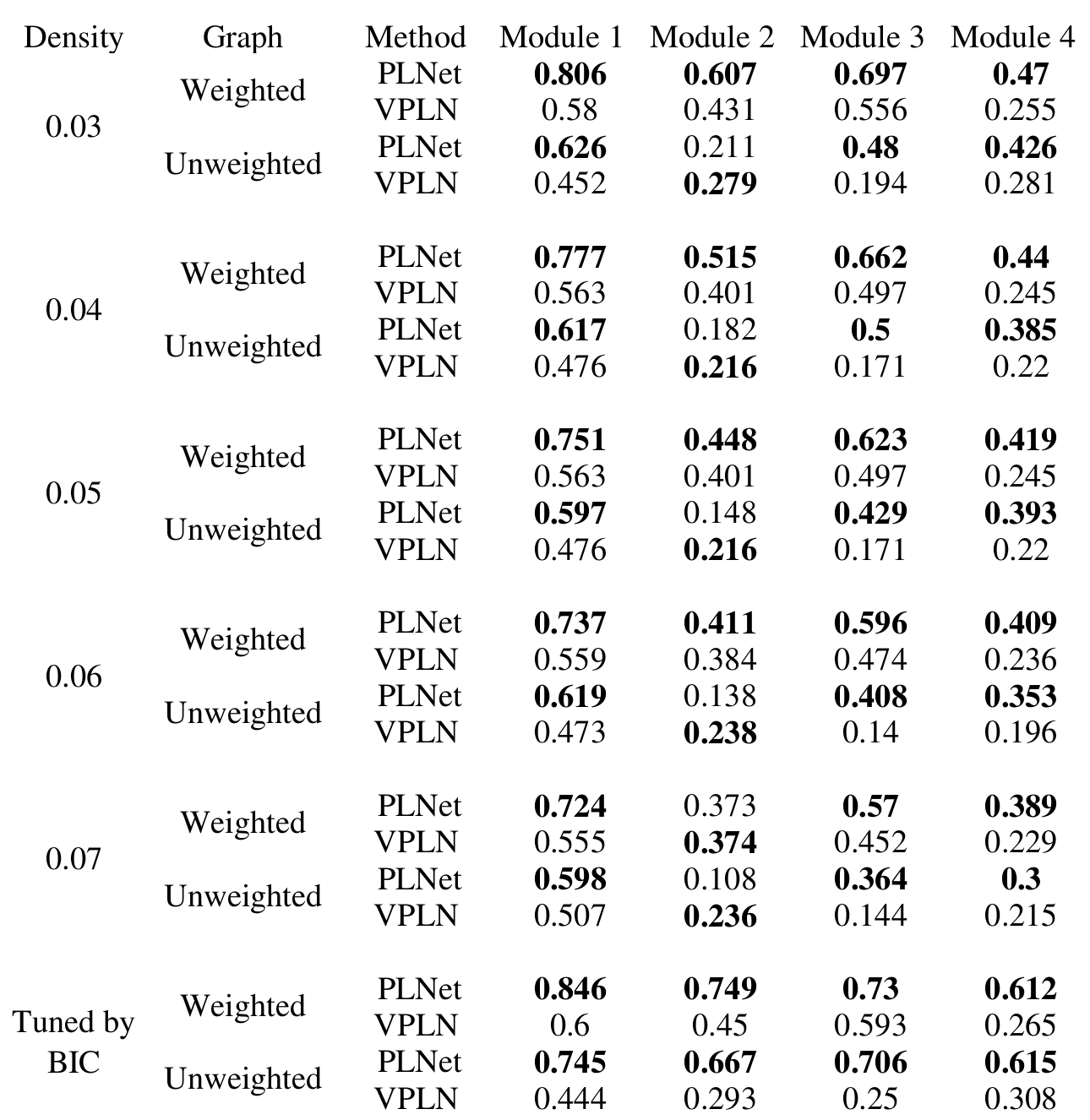}
	
	\label{tab5}
\end{table}	

\newpage
\setstretch{1.24}
\bibliographystyle{jasa}
\bibliography{PLND}
	
\end{document}